\documentclass[journal]{IEEEtran}

\usepackage{amsmath}
\usepackage{amsthm}
\usepackage{algorithm}
\usepackage{algorithmic}
\usepackage{array}
\usepackage{url}
\usepackage{makecell}

\usepackage{amsmath,amssymb,amsfonts,graphicx,nicefrac,mathtools,bbm}
\usepackage{amsthm}

\newcommand{\BIT}{\begin{itemize}}
\newcommand{\EIT}{\end{itemize}}
\newcommand{\BNUM}{\begin{enumerate}}
\newcommand{\ENUM}{\end{enumerate}}
 
\newcommand\mbb[1]{\mathbb{#1}}
\newcommand\mbf[1]{\mathbf{#1}}

\def\mrm#1{\mathrm{#1}}
\def\reals{\mathbb{R}} 
\def\complex{\mathbb{C}} 
\renewcommand{\exp}[1]{\operatorname{exp}\left(#1\right)} 
\def\indic#1{\mbb{I}\left({#1}\right)} 
\providecommand{\argmax}{\mathop\mathrm{arg max}} 
\providecommand{\argmin}{\mathop\mathrm{arg min}}

\def\P{\mathbb{P}} 

\def\Cov{\mrm{Cov}} 

\def\Unif{\textnormal{Unif}}

\newtheorem{theorem}{Theorem}
\newtheorem{lemma}{Lemma}
\newtheorem{assumption}{Assumption}

\newcommand{\algorithmicbreak}{\textbf{break}}
\newcommand{\BREAK}{\STATE \algorithmicbreak}
\newcommand{\algorithmiccontinue}{\textbf{Continue}}
\newcommand{\CONTINUE}{\STATE \algorithmiccontinue}

\newcommand{\Vabc}{\Delta\mathbf{V}}
\newcommand{\Iabc}{\Delta\mathbf{I}}
\newcommand{\Sabc}{\Delta\mathbf{S}}
\newcommand{\Yabc}{\mathbf{Y}}
\newcommand{\Vang}{\Delta \boldsymbol{\theta}}

\newcommand{\vabc}{\Delta\mathbf{v}}

\newcommand{\Vpnz}{\Delta\mathbf{V}^{pnz}}
\newcommand{\Ipnz}{\Delta\mathbf{I}^{pnz}}

\newcommand{\Ypnz}{\mathbf{Y}^{pnz}}

\newcommand{\vpnz}{\Delta\mathbf{v}^{pnz}}

\newcommand{\Bzero}{\mathbf{0}}

\usepackage{graphicx,listings,psfrag,amsfonts,cases}

\usepackage{pst-sigsys}
\usepackage{pstricks,pst-node,pst-tree}

\usepackage{arydshln}
\usepackage{multirow}

\usepackage{indentfirst}

\newcommand{\putFig}[3]{
        \begin{figure}[htbp] 
 		\centering
 		\includegraphics[width=#3]{#1}
		  \caption{#2}
                \label{fig:#1}
        \end{figure} }

\newcommand{\pa}[1]{\text{pa}(#1)}
\newcommand{\ga}[1]{\text{pa}(\text{pa}(#1))}
\newcommand{\si}[1]{\mathcal{S}(#1)}

\usepackage[american]{circuitikz}
\usetikzlibrary{calc}

\usepackage{color}
\newcommand{\revv}[1]{#1}

\begin{document}
	
\title{Unbalanced Multi-Phase Distribution Grid Topology Estimation and Bus Phase Identification}

\author{Yizheng~Liao,~\IEEEmembership{Student Member,~IEEE,}
        Yang~Weng,~\IEEEmembership{Member,~IEEE,}
        Guangyi~Liu,~\IEEEmembership{Senior Member,~IEEE,}
        Zhongyang~Zhao,
        Chin-Woo~Tan,
        Ram~Rajagopal,~\IEEEmembership{Member,~IEEE}
        \vspace{-4ex}
\thanks{Y. Liao, C-W.Tan, R. Rajagopal are with Department of Civil and Environmental
Engineering, Stanford University, Stanford, CA, 94305 USA e-mail: (\{yzliao, tancw,
ramr\}@stanford.edu). Y. Weng is with School of Electrical, Computing, and
Energy Engineering, Arizona State University, Tempe, AZ, 85287 USA e-mail:
yang.weng@asu.edu. G. Liu and Z. Zhao are with GEIRI
North America, San Jose, CA, 95134, USA e-mail: (guangyi.liu@geirina.net, ecezhao@gmail.com)}}
\maketitle

\begin{abstract}
There is an increasing need for monitoring and controlling uncertainties brought by distributed energy resources in distribution grids. For such goal, accurate multi-phase topology is the basis for correlating measurements in unbalanced distribution networks. Unfortunately, such topology knowledge is often unavailable due to limited investment, especially for \revv{low-voltage} distribution grids. Also, the bus phase labeling information is inaccurate due to human errors or outdated records. For this challenge, this paper utilizes smart meter data for an information-theoretic approach to learn the topology of distribution grids. Specifically, multi-phase unbalanced systems are converted into symmetrical components, namely positive, negative, and zero sequences. Then, this paper proves that the Chow-Liu algorithm finds the topology by utilizing power flow equations and the conditional independence relationships implied by the radial multi-phase structure of distribution grids with the presence of incorrect bus phase labels. At last, by utilizing Carson's equation, this paper proves that the bus phase connection can be correctly identified using voltage measurements. For validation, IEEE systems are simulated using three real data sets. The simulation results demonstrate that the algorithm is highly accurate for finding multi-phase topology even with strong load unbalancing condition and DERs. This ensures close monitoring and controlling DERs in distribution grids.
\end{abstract}

\maketitle

%
%

\section{Introduction}
The power distribution system is currently undergoing a dramatic transformation in \revv{both forms and functions. Large-scale deployments of technologies such as rooftop solar, electric vehicles (EVs), and smart home management systems have the potential to offer cheaper, cleaner and more controllable energy to the customers.} On the other hand, \revv{the integration of these resources} has been proven to be nontrivial, largely because of their inherent uncertainty and distributed nature.

For example, even a small-scale of distributed energy resources (DERs) can affect the stability of distribution grids \cite{dey2010urban}. Such a problem will be aggravated by the unbalance situation in distribution grids especially when uneven DER deployment happens. Furthermore, the more frequent bi-directional power flows easily leave the existing monitoring system with passive protective devices insufficient for robust grid operations. In addition to the static connectivity, mobile components, such as EVs, can further jeopardize the grid stability due to their frequent plug-in \cite{clement2010impact}. Therefore, the multi-phase grid monitoring tools need to be carefully designed for islanding and line work hazards in system operation with deep and uneven DER penetrations. For such monitoring, grid topology information is a prerequisite.

For topology estimation, the transmission grid assumes \revv{a prior knowledge of grids}, which needs limited error correction. Also, it is assumed that infrequent reconfiguration happens, identifiable by generalized state estimation \cite{huang2012electric, abur2004power, Lugtu80}. Unfortunately, such assumptions do not hold in medium- and low-voltage distribution grids, where topology can change relatively more frequently with limited sensing devices.  Furthermore, many urban distribution lines have been underground for decades, making prior knowledge of topology suspicious and expensive to verify \cite{rudin2012machine}. 

For distribution grid topology identification, many methods have been proposed in recent years. For example, in \cite{deka2015structure}, the correct topology is searched from a set of possible radial networks. Given the line parameters, \cite{cavraro2017voltage} and \cite{sharon2012topology} propose maximum likelihood methods to select the operational distribution grid topology. \cite{bolognani2013identification}, \cite{peppanen2016distribution}, and \cite{liao2018urban} utilize the statistical correlation of single-phase voltages collected from smart meters to estimate distribution grid topology. Unfortunately, all of these methods focus on the balanced or single-phase systems. For utility practice, distribution grids for buildings and residential areas are usually \revv{unbalanced and multi-phase systems}. One reason is that the loads connected at different phases are unbalanced due to the uneven growth in each feeder territory \cite{lo1993decomposed,tleis2007power}. For example, surveyed by the American National Standards Institute (ANSI), \revv{$2$\% of distribution grids in the USA} have a significant undesirable degree of unbalance \cite{ansc95,routtenberg2015pmu}. With the growth of renewable penetration, the load unbalance problem will become more frequent in future distribution grids. For example, the unbalance of the multi-phase system appears more often because the installations and operations of many DER devices are not fully controlled by utilities. This fact makes the requirement of balanced grids in previous works invalid in field applications.

In order to find the topology of unbalanced multi-phase distribution grids, \cite{yuan2016inverse} and its follow-up work \cite{ardakanian2017event} formulate multi-phase measurements as vectors and apply the single-phase approach to estimate grid topology. In \cite{deka2018topology}, the multi-phase power flow equations are linearized and the topology estimation is formulated as a statistical learning problem. For all these approaches, a prerequisite is \revv{installations of Phase Measurement Units (PMUs), which have not been widely available in distribution grids.} In addition, these methods assume bus phase labelings are correct at each bus. For many utilities, as high as $10\%$ phase labelings are incorrect or unknown because of human errors or outdated records. This high error rate makes identifying new topology based on existing methods not sound anymore.

For resolving the problems above, this paper proposes a data-driven method that utilizes the smart meter data in different phases to estimate the topology of multi-phase distribution grid systems. Building on our previous works on \revv{the probabilistic graphical model formulation} of distribution grids \cite{weng2017distributed}, firstly, this paper expands the method from \revv{the single-phase representation to multi-phase balanced systems} with incorrect bus phase labels. In such model, a node represents the multi-phase bus voltages and an edge between nodes indicates the statistical dependency among multi-phase bus voltage measurements. 

Subsequently, the system of three unbalanced phasors is converted to three symmetrical components, namely the positive, negative, and zero sequences. Then, the Chow-Liu algorithm is proved to be optimal for identifying the multi-phase topology by utilizing power flow equations and the conditional independence relationships implied by the radial multi-phase structure of distribution grids. As a highlight, the proposed method does not require PMUs and is robust to incorrect phase labels, which is a critical problem in distribution grid operations. This feature is due to the label-invariant property of mutual information. Another major contribution is bus phase correction and identification. Specifically, a data-driven approach is proposed to identify true bus phase connections by utilizing Carson's equation \cite{kersting2006distribution}, which is employed for deriving the primitive phase impedances of different lines.


The performance of the proposed method is verified by simulations on the \revv{IEEE $37$-bus, $123$-bus, and $8500$-bus distribution test cases \cite{kersting2001radial}.} Three different data sets are used for simulation: North California PG\&E residential household data sets, ADRES project data set \cite{Einfalt11,VUT16} that contains $30$ houses load profiles in Upper-Austria, and Pecan Street data set, which contains load data of $345$ houses with PV panels in Austin, Taxes. Simulations are conducted via GridLAB-D, an open source distribution grid simulator \cite{chassin2008gridlab} for multi-phase systems. \revv{Simulation results show that, provided with hourly measurements, the proposed algorithm perfectly estimates the topology of multi-phase distribution grids with noiseless measurements.}

The rest of the paper is organized as follows: Section~\ref{sec:model} introduces the modeling of the multi-phase distribution system and the problem of data-driven topology  estimation. Section~\ref{sec:alg} firstly proves the topology estimation problem of a multi-phase distribution grid can be solved as a mutual information maximization problem and proposes an algorithm to solve such a maximization problem in multi-phase setup. Also, a method is proposed to identify the bus phase connection. In Section~\ref{sec:alg_unbalance}, to address the unbalance in distribution grids, an unbalanced distribution grid is transformed to a symmetric system using sequence component frame and prove that the mutual information approach can still apply to grid topology estimation and phase identification. Section~\ref{sec:sim} evaluates the performance of our method using IEEE test cases and real data collected from different regions. Section~\ref{sec:con} concludes this paper. 

\section{Multi-phase Distribution Grid Modeling and Problem Formulation}
\label{sec:model}
A distribution grid is modeled by a graph $\mathcal{G} = (\mathcal{M},\mathcal{E})$, where the vertex set $\mathcal{M} = \{0,1,2,\cdots,M\}$ represents the set of buses and the unidirectional edge set $\mathcal{E} = \{(i,k), i,k \in \mathcal{M}\}$ represents the branches. The branch between two buses is not necessary to be multi-phase. In the distribution grid, bus $0$ is the substation with a fixed voltage and is the root of the tree graph. $\mathcal{M}^+$ denotes the set of buses excluding the substation, i.e., $\mathcal{M}^+ = \mathcal{M}\backslash \{0\}$. If bus $i$ and bus $k$ are connected, i.e., $(i,k) \in \mathcal{E}$, and bus $i$ is closer to the root (substation) than bus $k$, bus $i$ is the \textit{parent} of bus $k$ and bus $k$ is the \textit{child} of $i$. Let $\text{pa}(i)$ denote the parent bus of bus $i$. The root has no parent and all other buses in $\mathcal{M}^+$ have exactly one parent. Let $\mathcal{C}(i)$ denote the set of child buses of bus $i$ and use $\mathcal{S}(i) = \mathcal{C}(\text{pa}(i)) \backslash \{i\}$ to denote the set of sibling buses of bus $i$.

Let $a$, $b$, and $c$ denote the three phases of the distribution grid. The vector $\mathbf{V}^{abc}_i = [V^a_i, V^b_i, V^c_i]^T \in \complex^3$ denotes the nodal voltages at bus $i$, where $V^\phi_i$ denotes the line-to-ground complex voltage on phase $\phi$ and $T$ is the transpose operator. Similarly, $\mathbf{I}^{abc}_i = [I^a_i, I^b_i, I^c_i]^T \in \complex^3$ and $\mathbf{S}^{abc}_i = [S^a_i, S^b_i, S^c_i]^T \in \complex^3$ denote the vectors of current injections and injected complex powers at bus $i$, respectively. If a bus is only connected with one or two phases, the quantities of the missing phase are zeros. For example, if bus $i$ does not have phase $c$, $V^c_i = 0$, $I^c_i = 0$, and $S^c_i = 0$. For convenience, $\mathbf{V}_i$, $\mathbf{I}_i$, and $\mathbf{S}_i$ are used as the general notation of multi-phase quantities at bus $i$.

If bus $i$ and bus $k$ are connected, i.e., $(i,k) \in \mathcal{E}$, the relationship between their nodal voltages and currents can be expressed as follows \cite{laughton1968analysis,chen1991distribution}:
\begin{equation}
	\begin{bmatrix}
		\mathbf{I}_i \\ \mathbf{I}_k 
	\end{bmatrix}
	=
	\begin{bmatrix}
		\Yabc_{ik}+\frac{1}{2}\mathbf{B}_{i,\text{shunt}} & -\Yabc_{ik} \\
		-\Yabc_{ik} & \Yabc_{ik}+\frac{1}{2}\mathbf{B}_{k,\text{shunt}}
	\end{bmatrix}
	\begin{bmatrix}
		\mathbf{V}_i \\ \mathbf{V}_k
	\end{bmatrix},
\end{equation}
where $\Yabc_{ik} \in \complex^{3\times 3}$ denotes the admittance submatrix between bus $i$ and bus $k$ and $\mathbf{B}_{i,\text{shunt}} \in \complex^{3\times 3}$ denotes the shunt capacitance at bus $i$. In a multi-phase system, $\Yabc_{ik}$ is not diagonal. The voltages at different phases are coupled. As shown in Section~\ref{sec:sim}, this coupling property in multi-phase systems leads the existing single-phase methods to have poor performance in unbalanced multi-phase systems. Since $\mathbf{B}_{i,\text{shunt}}$ is relatively small in distribution grids \cite{kersting2012distribution}, $\mathbf{B}_{i,\text{shunt}}$ is assumed to be zeros, i.e.,  $\Bzero$. In the formulation above, the effect of the neural wire is merged into the multi-phase wires by applying Kron's reduction \cite{chen1991distribution}. If bus $i$ and bus $k$ are not connected, $\Yabc_{ik} = \mathbf{0}$. 

For bus $i$, the voltage measurement at time $n$ is $\mathbf{v}_i[n] = [v^a_i[n], v^b_i[n], v^c_i[n]]^T$, where $v^\phi_i[n] = |v^\phi_i[n]|\exp{j\theta^\phi_i[n]}$ denotes the complex voltage measurement on phase $\phi$ at time $n$ and $j=\sqrt{-1}$. The magnitude $|v^\phi_i[n]| \in \reals$ is in volt and the phase angle $\theta^\phi_i[n] \in \reals$ is in degree. All measurements are assumed to be noiseless at first. In Section~\ref{sec:sim}, the proposed algorithm will be validated with noisy measurements. In the following part, the upper-case letter denotes the symbol and the lower-case letter denotes the snapshot of symbol measurement. For example, $\mathbf{V}$ denotes the voltage symbol and $\mathbf{v}[n]$ denotes the voltage measurement at time $n$.

\revv{
With the modeling above, the multi-phase distribution grid topology estimation and bus phase identification problem is defined as
\begin{itemize}
	\item Problem: data-driven multi-phase distribution grid topology and bus phase estimation using voltage measurements
	\item Given: the time-series voltage measurements with unknown bus phase labels $\mathbf{v}_i[n]$, $n = 1,\cdots,N, i \in \mathcal{M}^+$
	\item Find: the unknown grid topology $\mathcal{E}$ and bus phase $\phi$.
\end{itemize}
}

\section{Multi-Phase Distribution Grid Topology Estimation and Bus Phase Identification}
\label{sec:alg}

This section firstly extends our previous work \cite{weng2017distributed} to estimate the topology of multi-phase system with incorrect phase labels. \revv{Then, a novel method is proposed to identify the true bus phase labels by utilizing the statistical relationship between voltage measurements.} The method proposed in this section focuses on balanced distribution systems. When a distribution system is unbalanced, a modified algorithm is proposed in Section~\ref{sec:alg_unbalance}. Fig.~\ref{fig:flowchart3} summarizes the criteria for the topology estimation method selection. 

\putFig{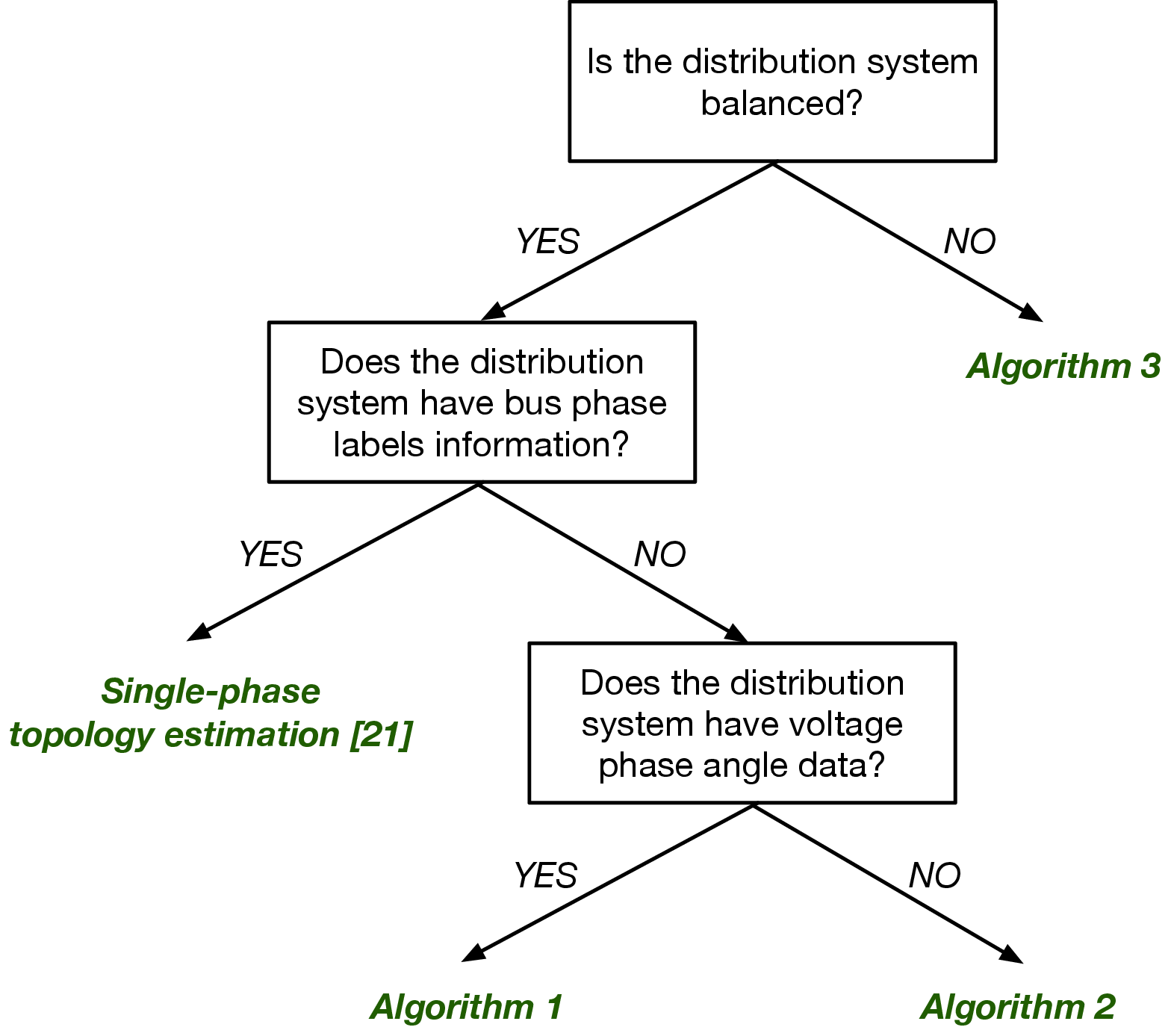}{Flow chart of topology estimation method selection.}{1\linewidth}

The end-user measurements are time-series data. One way to represent these data is using a probability distribution. If the nodal multi-phase voltage vector $\mathbf{V}_i$ is \revv{modeled as a random vector}, the joint distribution of voltage measurement $P(\mathbf{V}_{\mathcal{M}^+})$ is $P(\mathbf{V}_1) P(\mathbf{V}_2|\mathbf{V}_1)\cdots P(\mathbf{V}_M|\mathbf{V}_1,\cdots,\mathbf{V}_{M-1})$. Bus $0$ is omitted because it is the slack bus with a fixed voltage. 

\revv{Many previous works of distribution grid topology estimation \cite{deka2015structure,bolognani2013identification,liao2016urbanpes, weng2017distributed} only require the single-phase voltages. However, with the presence of false or unknown phase labels, all three phases voltage measurements are needed for topology estimation. The latter part of this section will show that our method is invariant to phase label accuracy and therefore can estimate topology with false or unknown phase labels.}

In many medium- and low-voltage distribution grids, the probability distribution of voltage is irregular. To better formulate the topology estimation problem, the incremental change of measurements is adopted in this paper \cite{chen2016quickest, liao2018urban, deka2016estimating}. At bus $i$, the incremental change of voltage is $\vabc_i[n] = \mathbf{v}^{abc}_i[n] - \mathbf{v}^{abc}_i[n-1]$ for $n \geq 2$. When $n=1$, $\vabc_i[1] = 0$. By using the incremental change $\Vabc$, the joint probability is
\begin{eqnarray}
	  P(\Vabc_{\mathcal{M}^+}) &=& P(\Vabc_1) P(\Vabc_2|\Vabc_1)\cdots \nonumber \\
	  && \times P(\Vabc_M|\Vabc_1,\cdots,\Vabc_{M-1}).
\end{eqnarray}

Since the nodal voltages are modeled as random vectors, the graph $\mathcal{G}$ becomes a probabilistic graphical model with a tree structure. In a graphical model, the vertex represents a random vector (e.g., $\Vabc_i$) and the edge between two vertices indicates the statistical dependency between bus voltages. Therefore, estimating distribution grid topology is equivalent to recovering the radial structure of the graphical model $\mathcal{G}$.

In a single-phase distribution grid, the nodal voltages only have statistical dependency with the nodal voltages of their parent bus \cite{weng2017distributed}. In the next part, such dependency will be extended from single-phase systems to approximate multi-phase systems' joint probability $P(\Vabc_{\mathcal{M}^+})$ as
\begin{equation}
\label{eq:approx}
	P(\Vabc_{\mathcal{M}^+}) \simeq \prod_{i=1}^M P(\Vabc_i|\Vabc_{\text{pa}(i)}).
\end{equation}
If (\ref{eq:approx}) holds, finding the structure of $\mathcal{G}$ is equivalent to finding the parent of each bus. \revv{The next part uses a two-stage approach to prove the approximation in (\ref{eq:approx}) holds with equality. In the first stage,  bus voltages are proved to be conditionally independent, given their parents, grandparents, and siblings, i.e., $P(\Vabc_{\mathcal{M}^+}) = \prod_{i=1}^M P(\Vabc_i|\Vabc_{\{\text{pa}(i), \text{pa}(\text{pa}(i)), \mathcal{S}(i)\}})$. Then, inspiring by the real data observation, (\ref{eq:approx}) is shown to holds with equality.}

Before starting the first stage proof, two assumptions are proposed and justified using real data.
\begin{assumption}
\label{assump:indept}
	In a multi-phase distribution gird,
	\begin{itemize}
		\item the incremental change of the current injection $\Iabc$ at each non-slack bus is independent, i.e., $\Iabc_i \perp \Iabc_k$ for all $i\neq k$.
		\item the incremental changes of the current injection $\Iabc$ and bus voltage $\Vabc$ at each bus follow Gaussian distribution with zero means and non-zero covariances. 
	\end{itemize}
\end{assumption}

Fig.~\ref{fig:IP_indept} shows the pairwise mutual information of the incremental changes of bus current injection using the real data from PG\&E. The mutual information $I(\mathbf{X}, \mathbf{Y})$ is a measure of the statistical dependence between two random vectors $\mathbf{X}$ and $\mathbf{Y}$. When the mutual information is zero, these two random vectors are independent, i.e., $\mathbf{X} \perp \mathbf{Y}$ \cite{cover2012elements}. In Fig.~\ref{fig:IP_indept}, most pairs of $\Iabc$ have small values. Thus, the current injections are assumed to be independent with some approximation errors. This assumption has also been adopted in other works, e.g., \cite{deka2015structure, bolognani2013identification, deka2018topology}. To further validate the independence of $\Iabc$, Fig.~\ref{fig:current_autocorr} plots the average auto-correlation of current injection increment of PG\&E data in the IEEE 123-bus system. The error bar is one standard deviation. In Fig.~\ref{fig:current_autocorr}, the auto-correlation of $\Iabc$ drops significantly as the lag increases. This observation justifies that the current injection increments are approximately independent over time.

\putFig{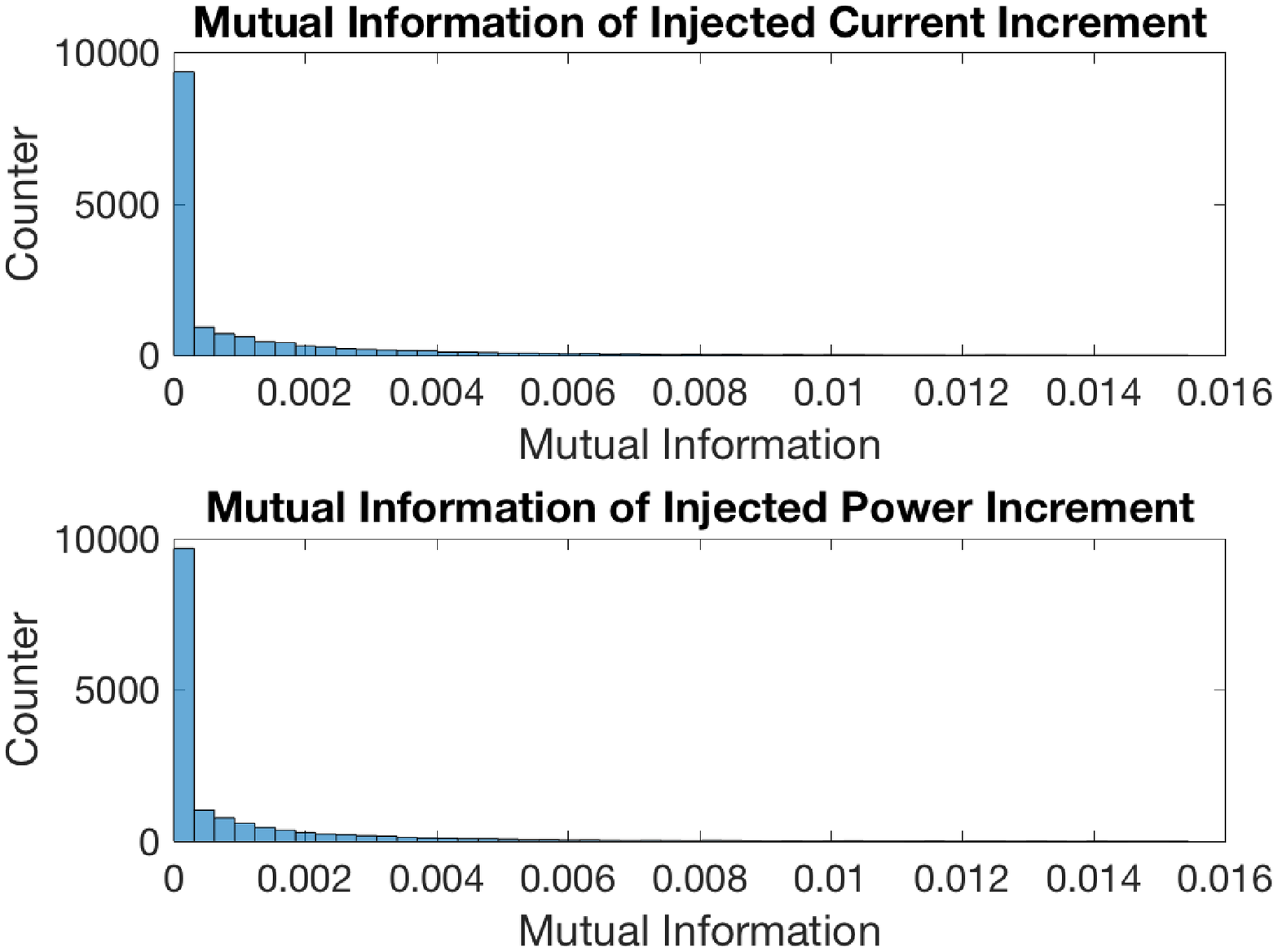}{Mutual information of pairwise current injection increment $\Iabc$ and power injection increment $\Sabc$ of PG\&E data sets in the IEEE 123-bus system.}{\linewidth}

\putFig{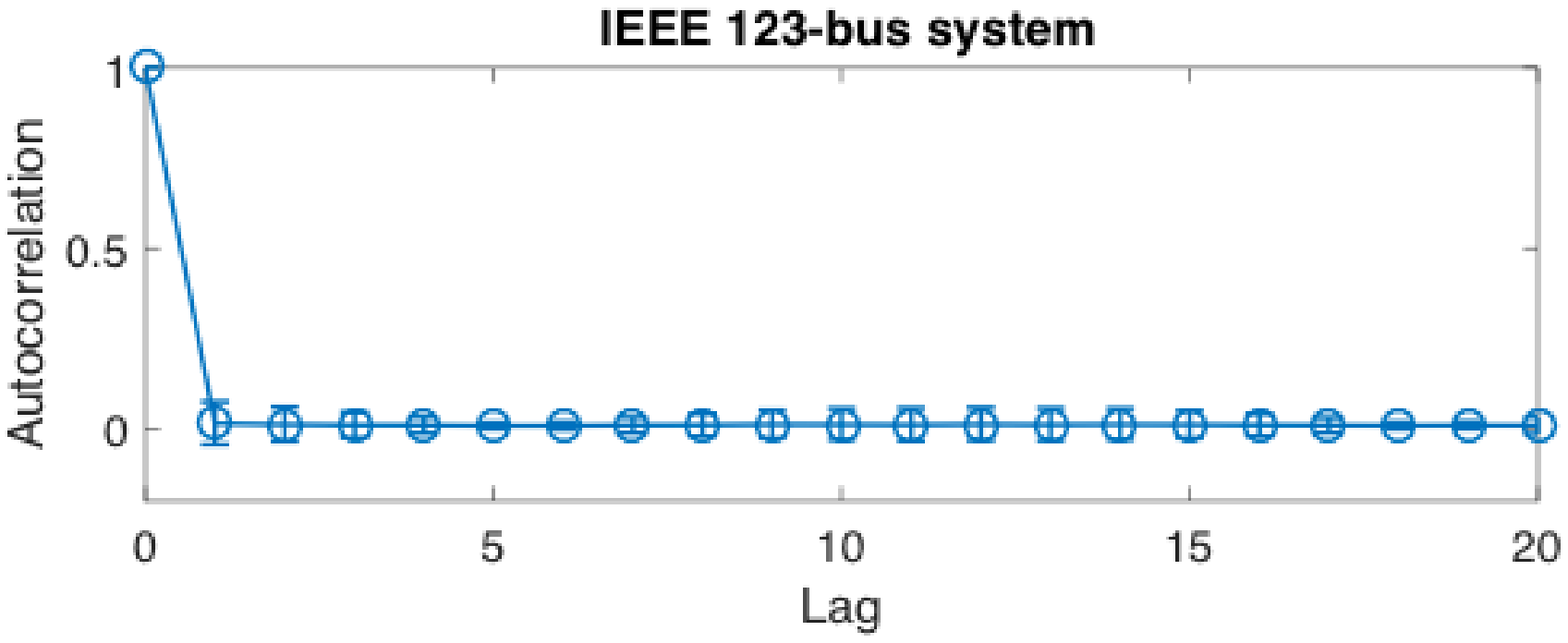}{Average auto-correlation of current injection increment $\Iabc$ of PG\&E data sets in the IEEE 123-bus system. The error bar is one standard deviation.}{\linewidth}

\revv{Both injected power increment independence and injected current increment independence are adopted in the existing works of distribution grid topology estimation. \cite{liao2018urban} uses the real data to show that these two assumptions are equivalent in distribution grids. Fig.~\ref{fig:IP_indept} illustrates the mutual information of pairwise power injection increment $\Sabc$ and pairwise current injection increment $\Iabc$.} Both histograms are similar. In this paper, the assumption of current injection independence is preferred because it simplifies the proof of following theorems and lemmas.

\putFig{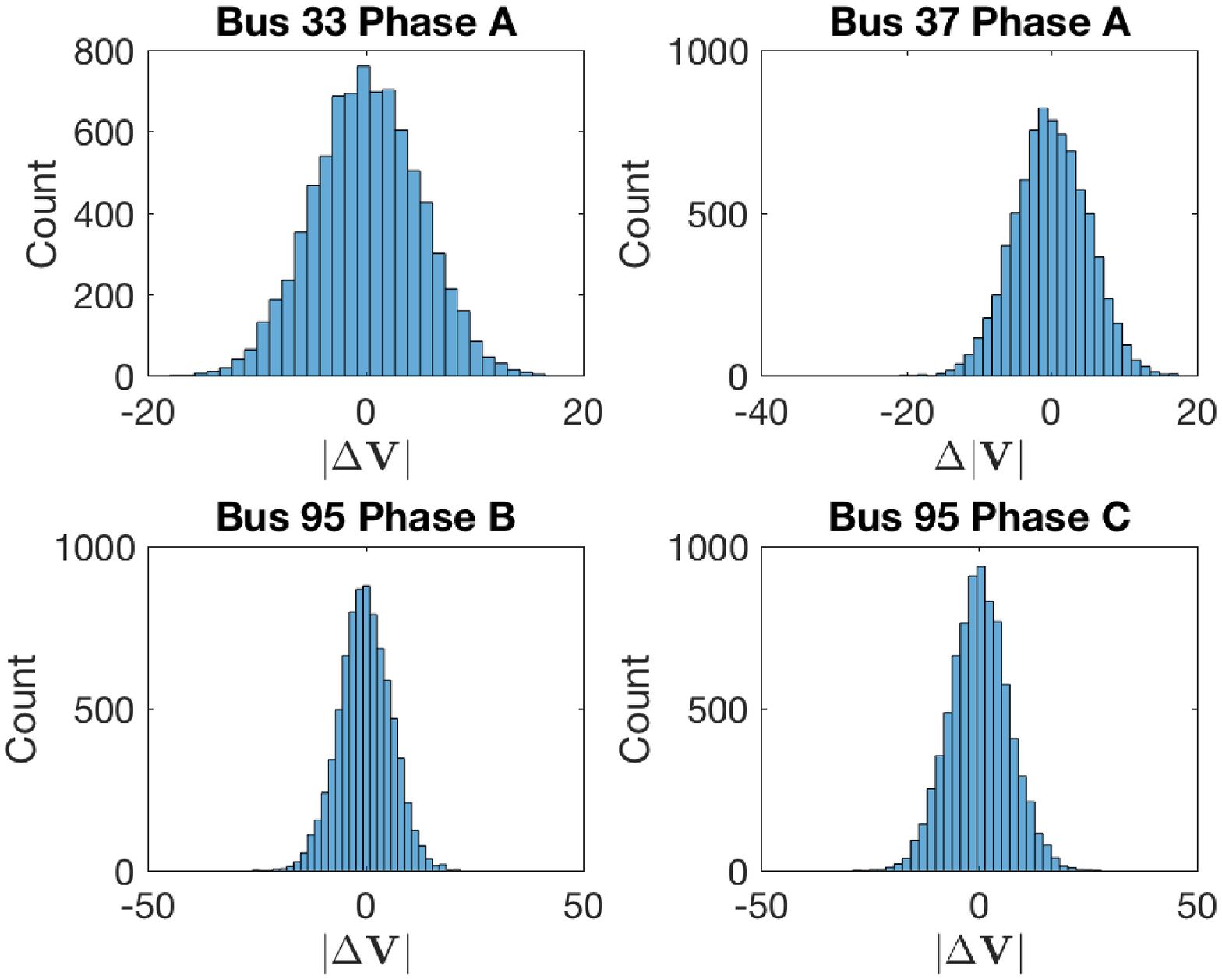}{Histograms of $|\Vabc|$ of four buses in IEEE 123-bus system using PG\&E data.}{\linewidth}

Fig.~\ref{fig:v_dist} illustrates the histograms of bus voltage $|\Vabc|$ in IEEE 123-bus system using PG\&E data. Hence, the voltage data approximately follow Gaussian distributions. With Assumption~\ref{assump:indept}, $P(\Vabc_{\mathcal{M}^+})$ is proved to be $\prod_{i=1}^M P(\Vabc_i|\Vabc_{\{\text{pa}(i), \text{pa}(\text{pa}(i)), \mathcal{S}(i)\}})$. For connivance, let $\mathcal{C}^\infty(i)$ denote all buses that are below bus $i$. For example, in Fig.~\ref{fig:6bus_example}, $\mathcal{C}^\infty(1) = \{2,3,4,5,6,7\}$ and $\mathcal{C}^\infty(2) = \{4,5\}$.

 
\putFig{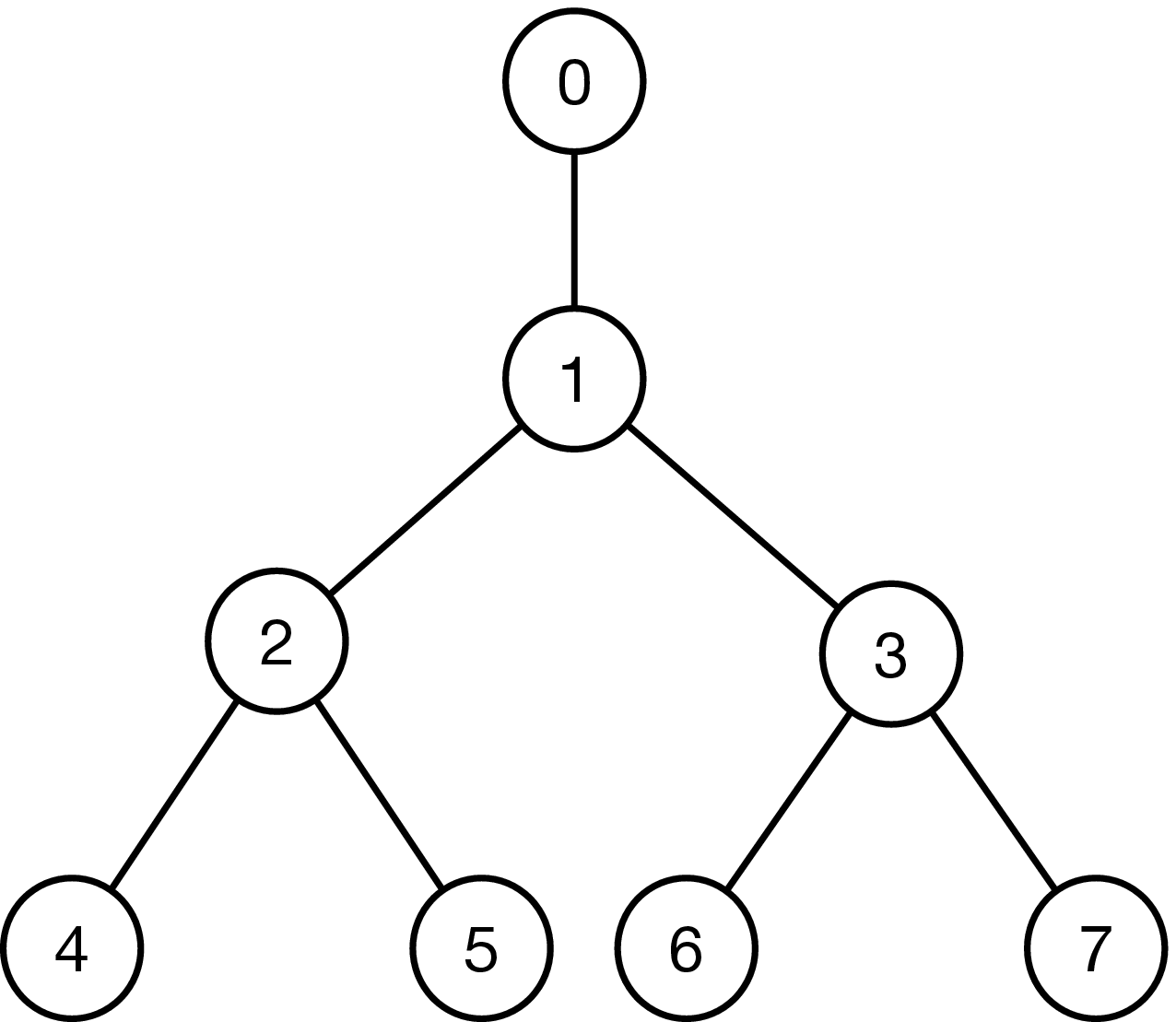}{An example of an $8$-bus multi-phase system. A node represents a bus, which can be single-phase or multi-phase. An edge represents a branch between two buses. The branch is unnecessary to be multi-phase. Bus $0$ is the substation (root).}{0.5\linewidth}
\begin{lemma}
\label{thm:cond_indept}
	If the incremental change of current injection at each bus is approximately independent (i.e., $\Iabc_i \perp \Iabc_k$ for $i\neq k$), given the incremental voltage changes of bus $i$'s parent ($\Vabc_{\text{pa}(i)}$), grandparent ($\Vabc_{\text{pa}(\text{pa}(i))}$), and siblings ($\Vabc_{\mathcal{S}(i)}$), the incremental voltage changes of bus $i$ and the buses that are not below bus $i$ are conditionally independent, i.e., $\Vabc_i \perp \Vabc_k | \Vabc_{\{\text{pa}(i), \text{pa}(\text{pa}(i)), \mathcal{S}(i)\}}$ for $k \notin \{\text{pa}(i), \text{pa}(\text{pa}(i)),\mathcal{S}(i), \mathcal{C}^\infty(i)\}$ and $i\neq k$.
\end{lemma}

Here, a simple example demonstrates Lemma~\ref{thm:cond_indept}. A formal proof is given in Appendix section ~\ref{sec:proof_cond_indept}. For the example system in Fig.~\ref{fig:6bus_example}, the nodal admittance equation is $\Yabc_{\mathcal{M}^+}\Vabc_{\mathcal{M}^+}=\Iabc_{\mathcal{M}^+}$, where
\begin{equation}
	\label{eq:nodal_eqn}
\Yabc_{\mathcal{M}^+} = 
\begin{bmatrix}
	\Yabc_{11} & \Yabc_{12} & \Yabc_{13} & \Bzero & \Bzero & \Bzero & \Bzero\\
	\Yabc_{21} & \Yabc_{22} & \Bzero & \Yabc_{24} & \Yabc_{25} & \Bzero & \Bzero\\
	\Yabc_{31} & \Bzero & \Yabc_{33} & \Bzero & \Bzero & \Yabc_{36} & \Yabc_{37}\\
	\Bzero & \Yabc_{42} &  \Bzero & \Yabc_{44} & \Bzero & \Bzero & \Bzero\\
	\Bzero & \Yabc_{52} &  \Bzero & \Bzero & \Yabc_{55} & \Bzero & \Bzero\\
	\Bzero & \Bzero &  \Yabc_{63} & \Bzero & \Bzero & \Yabc_{66} & \Bzero \\
	\Bzero & \Bzero &  \Yabc_{73} & \Bzero & \Bzero & \Bzero & \Yabc_{77}
\end{bmatrix},
\end{equation}
$\Yabc_{ik} = \Yabc_{ki}$, and $\Yabc_{ii} = - \sum_{k=0,k\neq i}^7 \Yabc_{ik}$. If $\Yabc_{ik} = \Bzero$, there is no branch between bus $i$ and $k$. 

For bus $4$, $\pa{4} = 2$, $\ga{4} = 1$, and $\si{4} = \{5\}$. Therefore, given $\Vabc_1 = \vabc_1$, $\Vabc_2 = \vabc_2$, and $\Vabc_5 = \vabc_5$, there are the following equations:
\begin{eqnarray}
	\Iabc_1 &=& \Yabc_{11}\vabc_1 + \Yabc_{12}\vabc_2 + \Yabc_{13}\Vabc_3, \label{eq:i1} \\
	\Iabc_4 &=& \Yabc_{42}\vabc_2 + \Yabc_{44}\Vabc_4,\label{eq:i4} \\
	\Iabc_6 &=& \Yabc_{63}\Vabc_3 + \Yabc_{66}\Vabc_6, \label{eq:i6}\\
	\Iabc_7 &=& \Yabc_{73}\Vabc_3 + \Yabc_{77}\Vabc_7. \label{eq:i7}
\end{eqnarray}
Given $\Iabc_1 \perp \Iabc_4$, according to (\ref{eq:i1}) and (\ref{eq:i4}), $\Vabc_3$ and $\Vabc_4$ are conditionally independent given $\Vabc_1$, $\Vabc_2$, and $\Vabc_5$. In (\ref{eq:i1}), $\Vabc_3$ can be rewritten as $(\Yabc_{13})^{-1}(\Iabc_1 - \Yabc_{11}\vabc_1 - \Yabc_{12}\vabc_2)$. Then, $\Vabc_3$ is substituted into (\ref{eq:i6}). Since $\Iabc_1$, $\Iabc_4$, and $\Iabc_6$ are independent, $\Iabc_4$ and $\Iabc_6 - \Yabc_{63}(\Yabc_{13})^{-1}\Iabc_1$ are independent. Therefore, $\Vabc_4$ and $\Vabc_6$ are conditionally independent. Similarly, $\Vabc_4$ and $\Vabc_7$ are conditionally independent.

For a non-leaf bus, bus $2$, given $\Vabc_3 = \vabc_3$ and $\Vabc_1 = \vabc_1$, there are the following equations:
\begin{eqnarray}
    \Iabc_1 &=& \Yabc_{11}\vabc_1 + \Yabc_{12}\Vabc_2 + \Yabc_{13}\vabc_3, \label{eq:i11} \\
    \Iabc_6 &=& \Yabc_{63}\vabc_3 + \Yabc_{66}\Vabc_6, \label{eq:i66} \\
    \Iabc_7 &=& \Yabc_{73}\vabc_3 + \Yabc_{77}\Vabc_7. \label{eq:i77}
\end{eqnarray}
Given $\Iabc_1$ and $\Iabc_6$ are independent, according to (\ref{eq:i11}) and (\ref{eq:i66}), $\Vabc_2$ and $\Vabc_6$ are conditionally independent. Similarly, $\Vabc_2$ and $\Vabc_6$ are conditionally independent given $\Iabc_1$ and $\Iabc_7$ are independent. Our conclusion in Lemma~\ref{thm:cond_indept} is similar to the results in \cite{deka2017topology}.



\begin{assumption}
\label{ass:ass2}
	In a distribution grid, the mutual information between $\Vabc_i$ and its parent $\Vabc_{\pa{i}}$ is much larger than the mutual information between $\Vabc_i$ and $\Vabc_{\ga{i}}$ and $\Vabc_{\si{i}}$.
\end{assumption}

\putFig{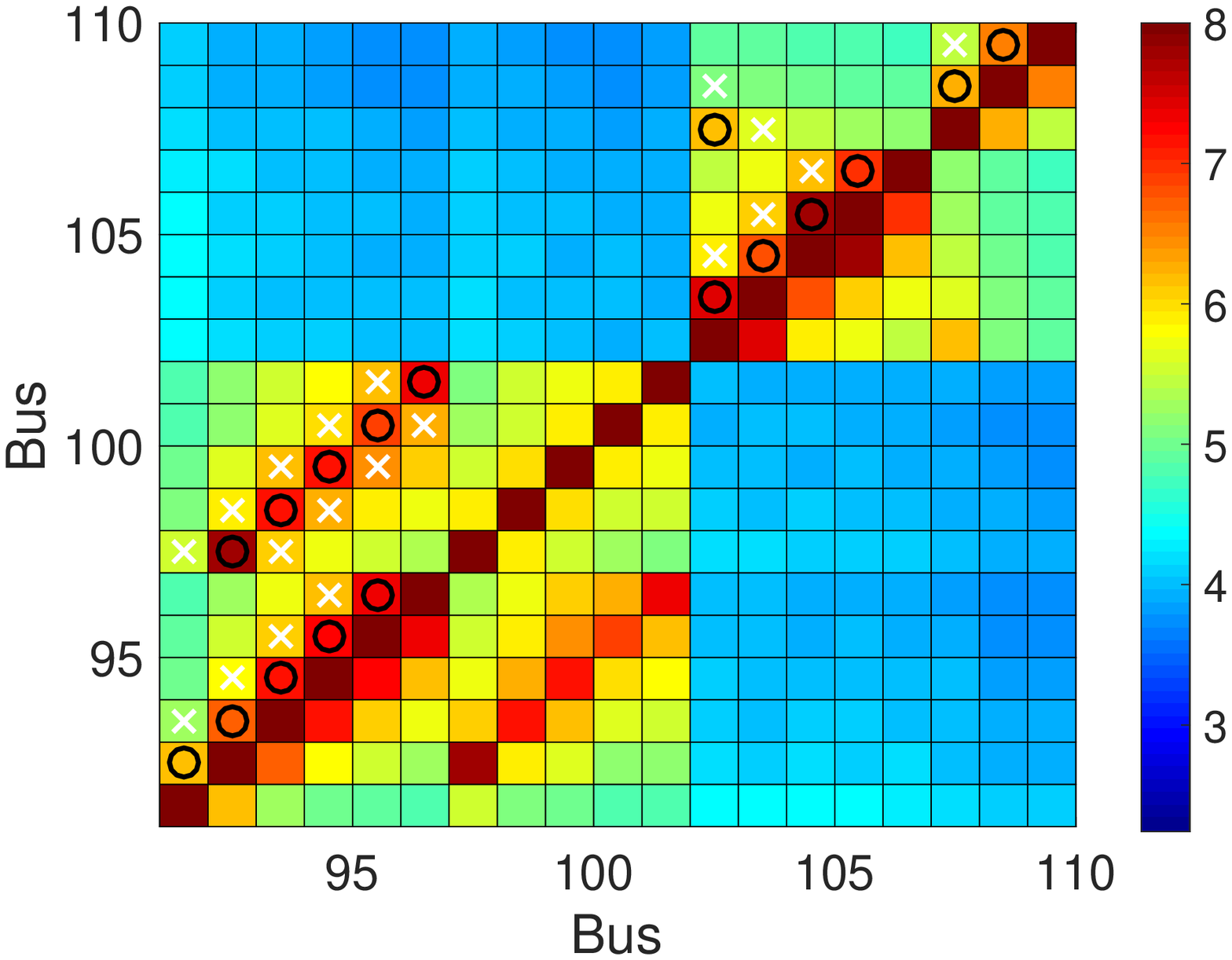}{Mutual information of pairwise buses in IEEE 123-bus system using PG\&E data sets. The circle indicates the neighbors of bus $i$. The crossing indicates the two-step neighbor of bus $i$. The square without markers represents the bus pair that are more than two-step away.}{\linewidth}

Assumption~\ref{ass:ass2} is inspired by \revv{the real data observations}. Fig.~\ref{fig:bus_corr} plots the mutual information of voltage increments between each bus pair in \revv{IEEE 123-bus distribution system using the PG\&E data}. The distribution grid configuration and simulation setup are described in Section~\ref{sec:sim}. In Fig.~\ref{fig:bus_corr}, the color in a square represents the \revv{mutual information of voltage increments between two buses}. \revv{If the voltage increments of two buses are independent, their mutual information is zero\cite{cover2012elements} (dark color)}. In Fig.~\ref{fig:bus_corr}, the circle refers to the bus neighbors (e.g. parent bus) and the crossing indicates the two-step neighbors (grandparent bus and sibling buses). If a square does not have any marker, the corresponding pair of buses is more than two-step away. In Fig.~\ref{fig:bus_corr}, the mutual information between the voltages of two-step neighbors is higher than the mutual information of other bus pairs, but it is still much lower than the mutual information between two neighbors. The diagonal bus pairs have the highest mutual information because it is the self-information. With Assumption~\ref{ass:ass2}, $P(\Vabc_{\mathcal{M}_+})$ can be simplified to only depend on the voltages of parent buses. Section~\ref{sec:sim} uses numerical simulations to demonstrate that this approximation does not degrade the performance of topology and bus phase connectivity estimation. 

\begin{lemma}
\label{lemma:one_hop_indept}
Given the incremental voltage changes of bus $i$ in a multi-phase distribution grid, if the incremental change of current injection at each bus is approximately independent, the incremental voltage changes of every pair of bus $i$'s children are conditionally independent, i.e., $\Vabc_k \perp \Vabc_l | \Vabc_i $ for $k, l \in \mathcal{C}(i)$ and $k\neq l$.
\end{lemma}

With Lemma~\ref{lemma:one_hop_indept}, (\ref{eq:approx}) holds with equality, i.e., $P(\Vabc_{\mathcal{M}^+}) = \prod_{i=1}^M P(\Vabc_i|\Vabc_{\text{pa}(i)})$. Thus, finding the distribution grid topology is equivalent to finding the parent of each bus. In the following subsections, an information theoretical approach is proposed to estimate the multi-phase distribution grid topology with incorrect bus phase labels.

\subsection{An Information Theoretical Approach to Estimate Multi-phase Distribution Grid Topology}
One way to find the parent of each bus is minimizing the Kullback-Leibler divergence \cite{cover2012elements} of $P(\Vabc_{\mathcal{M}^+})$ and $Q(\Vabc_{\mathcal{M}^+}) = \prod_{i=1}^M P(\Vabc_i|\Vabc_{\text{pa}(i)})$, i.e.,
\begin{equation}
	\label{eq:KL}
	\widehat{\boldsymbol{\Theta}} = \argmin_{\boldsymbol{\Theta} \subset \mathcal{M}^+} D(P(\Vabc_{\mathcal{M}^+})\|Q(\Vabc_{\mathcal{M}^+};\boldsymbol{\Theta})),
\end{equation}
where $\boldsymbol{\Theta}$ denotes the collection of parent bus index of every bus, i.e., $\boldsymbol{\Theta} = \{\text{pa}(1),\cdots, \text{pa}(M)\}$, $P$ denotes the joint distribution of all voltages, and $Q$ denotes the distribution of voltage vectors with tree structure. When two distributions are identical, the KL divergence is zero. Therefore, as shown in Lemma~\ref{lemma:one_hop_indept}, if there exists a distribution $Q(\Vabc_{\mathcal{M}^+};\widehat{\boldsymbol{\Theta}})$ that is identical to $P(\Vabc_{\mathcal{M}^+})$, $\widehat{\boldsymbol{\Theta}}$ contains the parent bus index of every bus $i$. The associated structure of $P_{CL}(\Vabc_{\mathcal{M}^+}) = Q(\Vabc_{\mathcal{M}^+};\widehat{\boldsymbol{\Theta}})$ is the estimated topology of distribution grid. Lemma~\ref{thm:MI_sum} proves that (\ref{eq:KL}) can be efficiently solved by utilizing the radial structure of distribution grids. In the following context, $\boldsymbol{\Theta}_i$ and $\text{pa}(i)$ are used interchangeably. 
\begin{lemma}
\label{thm:MI_sum}
In a radial distribution grid, finding the topology is equivalent to solving the following optimization problem:
	\begin{equation}
	\label{eq:MI_sum}
		\widehat{\boldsymbol{\Theta}} = \argmax_{\boldsymbol{\Theta} \subset \mathcal{M}^+} \sum_{i=1}^M I\left(\Vabc_i; \Vabc_{\boldsymbol{\Theta}_i}\right),
	\end{equation}
	where $I\left(\Vabc_i; \Vabc_{\boldsymbol{\Theta}_i}\right)$ denotes the mutual information.
\end{lemma}

The proof is in Appendix~\ref{app:proof_MI_sum}. \revv{With Lemma~\ref{thm:MI_sum}}, a mutual information-based maximum weight spanning tree algorithm, well-known as Chow-Liu algorithm \cite{chow1968approximating}, could find $\widehat{\boldsymbol{\Theta}}$ and identify the multi-phase distribution grid topology. This algorithm has been applied to single-phase system in \cite{weng2017distributed}. Theorem~\ref{thm:chow-liu} proves that Chow-Liu algorithm can be extended to multi-phase systems.
\begin{theorem}
	\label{thm:chow-liu}
	In a radial multi-phase distribution grid, the mutual information-based maximum weight spanning tree algorithm (Chow-Liu algorithm) estimates the best-fitted topology.
\end{theorem}
\begin{proof}
This proof shows that the mutual information between connected buses is higher than those without a connection. If bus $i$ is the parent of bus $k$ and bus $l$ and $k\neq l$, by utilizing the chain rule property of the mutual information \cite{cover2012elements}, the joint mutual information is expressed as
\begin{eqnarray}
	&&I(\Vabc_i; \Vabc_k, \Vabc_l ) \nonumber \\
	&=& I(\Vabc_i;\Vabc_k) - I(\Vabc_i, \Vabc_k| \Vabc_l), \nonumber \\
	&=& I(\Vabc_k; \Vabc_l) - I(\Vabc_k, \Vabc_l|\Vabc_i).
\end{eqnarray}	
Since $\Vabc_k|\Vabc_i \perp \Vabc_l|\Vabc_i$, the conditional mutual information $I(\Vabc_k, \Vabc_l | \Vabc_i)$ is zero. Then
\begin{equation}
	I(\Vabc_i; \Vabc_k) = I(\Vabc_k; \Vabc_l) + I(\Vabc_i, \Vabc_k| \Vabc_l).
\end{equation}
Due to the fact that mutual information is always non-negative, $I(\Vabc_k; \Vabc_i) \geq I(\Vabc_k; \Vabc_l)$. Therefore, the mutual information between connected buses is larger than the mutual information between not connected buses. Then, by using the mutual information as the weight, the maximum weight spanning tree algorithm (Chow-Liu algorithm) solves (\ref{eq:MI_sum}) and estimates the distribution grid topology \cite{weng2017distributed, chow1968approximating}.
\end{proof}

The mutual information $I(\Vabc_i;\Vabc_k)$ can be computed as
\begin{equation}
	I(\Vabc_i;\Vabc_k) = H(\Vabc_i) + H(\Vabc_k) - H(\Vabc_i;\Vabc_k),\label{eq:mutual_info_entropy}
\end{equation}
where $H(\Vabc_i)$ denotes the entropy of $\Vabc_i$ and $H(\Vabc_i;\Vabc_k)$ denotes the cross-entropy of $\Vabc_i$ and $\Vabc_k$. An advantage using (\ref{eq:mutual_info_entropy}) is that many distributions have closed forms of entropy. In Assumption~\ref{assump:indept}, the incremental changes of voltages in distribution grids are assumed to follow Gaussian distribution approximately. Thus, the entropy of $\Vabc_i$ is 
\begin{equation}
\label{eq:gaussian_entropy}
H(\Vabc_i) = \frac{r}{2}\log(2\pi\exp{1}) + \frac{1}{2}\log(\det\Cov(\Vabc_i)),
\end{equation}
where $r$ denotes the dimension of the random vector $\Vabc_i$ and $\Cov$ denotes the covariance matrix. In some systems, the bus may not have all three phases. In this case, the disconnected phases are excluded in the computation of entropy.

A practical issue that exists in many distribution grids, especially the low-voltage distribution grids, is that the smart meter phase connectivity information is inaccurate. In some countries, about $10\%$ phase labels in low-voltage distribution grids are \revv{false or unknown}. Also, bus phase labels can change over time when new customers and DER devices are connected to grids \cite{wang2016phase}. As the correct bus phase connectivity information is critical to distribution grid plannings, the grid topology and phase connection should be estimated at the same time. \revv{To identify true bus phase labels, one may apply existing methods \cite{wang2016phase,short2013advanced,arya2011phase} to identify phase connectivity before estimating topology. Fortunately, our topology estimation method does not require this preprocessing step and is invariant to false phase labels.} Specifically, when voltage phases are incorrectly labeled, the elements in random vector $\Vabc_i$ are permuted. This permutation does not affect the computation of $\det\Cov(\Vabc_i)$, thus, does not change the values of $H(\Vabc_i)$ and $H(\Vabc_i; \Vabc_k)$. Therefore, the mutual information $I(\Vabc_i; \Vabc_k)$ is the same even the bus labels are incorrect. Section~\ref{sec:incorrect_phase} uses numerical examples to show that our algorithm can recover the topology perfectly with the presence incorrect phase labels.

\begin{algorithm}[htbp!]
\caption{Multiphase Distribution Grid Topology Estimation}
\label{alg:topology00}
\begin{algorithmic}[1]
\REQUIRE $\vabc_i[n]$ for $i \in \mathcal{M}^+$, $n = 1, \cdots N$
\FOR{$i,k \in \mathcal{M}^+$}
\STATE Compute empirical mutual information $I(\Vabc_i;\Vabc_k)$ based on $\vabc_i[n]$ and $\vabc_k[n]$ using (\ref{eq:mutual_info_entropy}) and (\ref{eq:gaussian_entropy}).
\ENDFOR
\STATE Sort all possible bus pair $(i,k)$ into non-increasing order by $I(\Vabc_i;\Vabc_k)$. Let $\mathcal{T}$ denote the sorted set.
\STATE Let $\widehat{\mathcal{E}}$ be the set of nodal pair comprising the maximum weight spanning tree. Set $\widehat{\mathcal{E}} = \emptyset$.
\FOR{$(i,k) \in \mathcal{T}$}
\IF {cycle is detected in $\hat{\mathcal{E}} \cup (i,k)$}
	\CONTINUE
\ELSE
	\STATE $\widehat{\mathcal{E}} \leftarrow \widehat{\mathcal{E}} \cup (i,k)$
\ENDIF
\IF {$|\widehat{\mathcal{E}}| == M$}
	\BREAK
\ENDIF
\RETURN $\widehat{\mathcal{E}}$
\ENDFOR
\end{algorithmic}
\end{algorithm}
The proposed algorithm for multi-phase distribution grid topology estimation is summarized in Algorithm~\ref{alg:topology00}. The well-known Kruskal's minimum weight spanning tree algorithm \cite{kruskal1956shortest,cormen2001introduction} can be applied to efficiently build the maximum weight spanning tree (Steps 6 - 16). The running time of the Kruskal's algorithm is $O(M\log M)$ for a radial distribution network with $M$ buses.

\subsection{Distribution Grid Topology Estimation using Voltage Magnitudes Only}
\label{sec:vmag}
Voltage phase angles are hard to acquire in distribution grids today because PMUs are not widely available. However, the proposed method can be extended to find the distribution grid topology only using voltage magnitudes $|\Vabc|$. As presented in Lemma~\ref{thm:MI_sum}, the key step of the proposed method is computing the mutual information of each bus voltage pair. Using chain rule, the mutual information $I(\Vabc_i; \Vabc_k)$ can be decomposed as
\begin{eqnarray}
	I(\Vabc_i; \Vabc_k) &=& I(|\Vabc_i|, \Vang_i; |\Vabc_k|, \Vang_k) \\
	&=& \underbrace{I(|\Vabc_i|, |\Vabc_k|)}_{\text{term A}} + \underbrace{I(|\Vabc_i|, \Vang_i \big| |\Vabc_k|)}_{\text{term B}} \nonumber \\
	&& + \underbrace{I(|\Vabc_i|, \Vang_k \big| |\Vabc_k|,\Vang_i)}_{\text{term C}}. \label{eq:breakdown}
\end{eqnarray}

\putFig{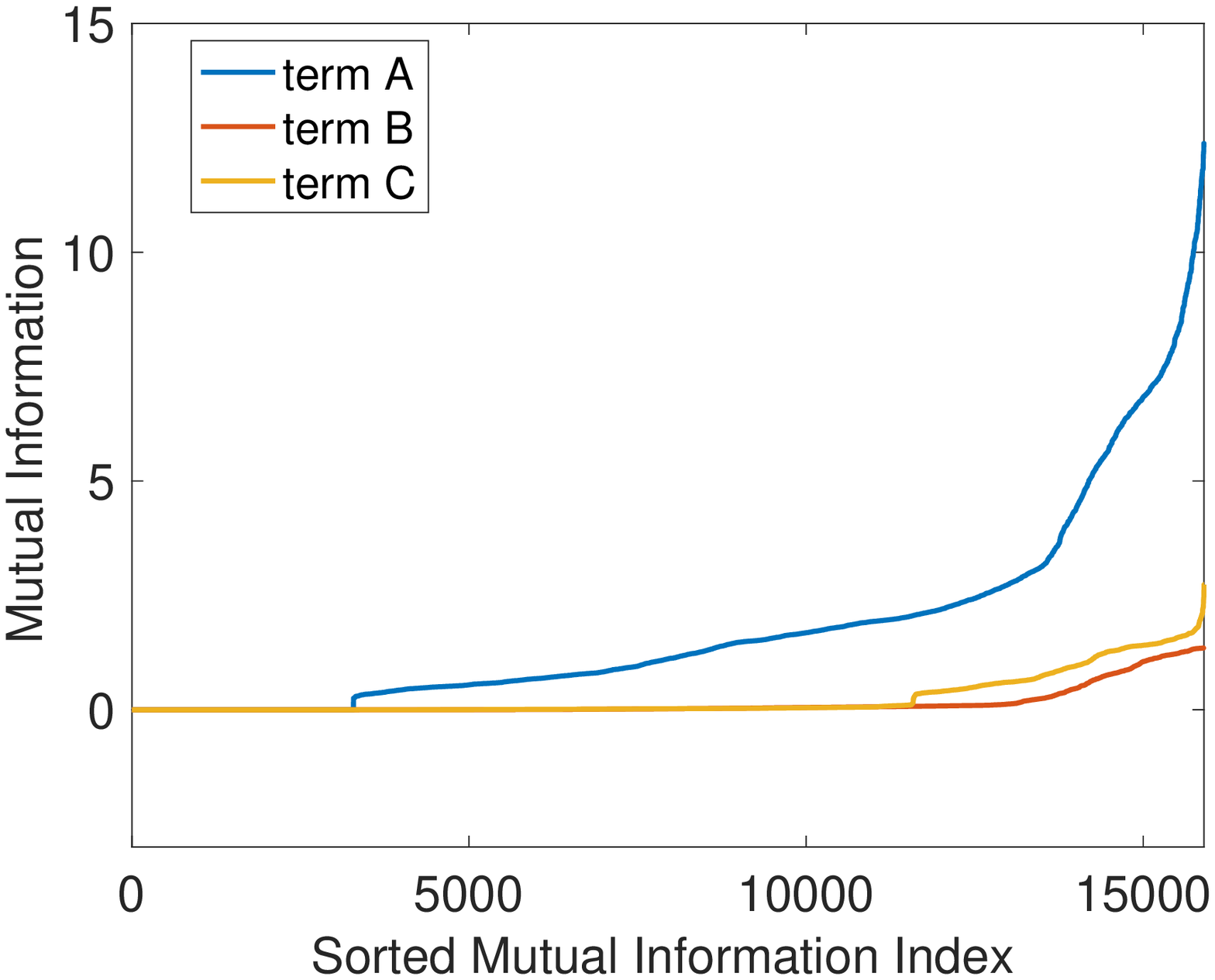}{Pairwise mutual information breakdown.}{0.9\linewidth}
Fig.~\ref{fig:MI_breakdown} empirically plots the pairwise mutual information of term A, B, C using the IEEE 123-bus system and the real data from PG\&E. The mutual information computed in (\ref{eq:breakdown}) is sorted by its value. The $x$-axis of Fig.~\ref{fig:MI_breakdown} is the index of the sorted mutual information. The $y$-axis is the mutual information of each part in (\ref{eq:breakdown}). The values of term A is much larger than term B and term C across all pairs of bus. The reason is that the changes of voltage angles are relatively small in distribution grids and thus, contain less information than voltage magnitudes. Based on our empirical observation in Fig.~\ref{fig:MI_breakdown}, $I(|\Vabc_i|; |\Vabc_k|)$ can be used to approximately estimate distribution grid structures. Specifically, the optimization problem in Lemma~\ref{thm:MI_sum} is approximated as
	\begin{equation}
	\label{eq:MI_sum_approx}
		\widehat{\boldsymbol{\Theta}} = \argmax_{\boldsymbol{\Theta} \subset \mathcal{M}^+} \sum_{i=1}^M I\left(|\Vabc_i|; |\Vabc_{\boldsymbol{\Theta}_i}\right|).
	\end{equation}
Many smart meters deployed can measure the voltages of all three phases. Therefore, the proposed algorithm can still apply when only smart meter voltage magnitude measurements are available. Algorithm~\ref{alg:topology_mag} summarizes the process for estimating topology using voltage magnitudes only.

\begin{algorithm}[htbp]
\caption{Multiphase Distribution Grid Topology Estimation using Voltage Magnitudes}
\label{alg:topology_mag}
\begin{algorithmic}[1]
\REQUIRE $|\vabc_i[n]|$ for $i \in \mathcal{M}^+$, $n = 1, \cdots N$
\FOR{$i,k \in \mathcal{M}^+$}
\STATE Compute empirical mutual information $I(|\Vabc_i|;|\Vabc_k|)$ based on $|\vabc_i[n]|$ and $|\vabc_k[n]|$ using (\ref{eq:mutual_info_entropy}) and (\ref{eq:gaussian_entropy}).
\ENDFOR
\STATE Sort all possible bus pair $(i,k)$ into non-increasing order by $I(|\Vabc_i|;|\Vabc_k|)$. Let $\mathcal{T}$ denote the sorted set.
\STATE Repeat Step 5 to Step 16 in Algorithm~\ref{alg:topology00}.
\end{algorithmic}
\end{algorithm}

\subsection{Topology Estimation of Weakly Mesh Distribution Grid}
In the previous part, the multi-phase distribution grid topology estimation method is proposed for the radial system. In practice, with the increase penetration of DERs, more distribution grids become to mesh structure for robustness \cite{liao2015distribution,liao2018urban,cavraro2019voltage}. In mesh structures, a bus has more than one parents. Assuming only one bus has two parents, the joint distribution is rewritten as
\begin{eqnarray}
	P(\Delta\mathbf{V}_{\mathcal{M}^+}) &=& P(\Vabc_M|\Vabc_{\text{pa}(M),1},\Vabc_{\text{pa}(M),2}) \nonumber \\
	&& \times \prod_{i=1}^{M-1}P(\Vabc_i|\Vabc_{\text{pa}(i)}),
\end{eqnarray}
where $\text{pa}(i),1$ and $\text{pa}(i),2$ represent the first and second parent of bus $i$. Following the same proof as Lemma~\ref{thm:MI_sum}, the KL distance $D(P(\Vabc_{\mathcal{M}^+})\|Q(\Vabc_{\mathcal{M}^+}))$ can be written as
\begin{eqnarray}
	&& D(P(\Vabc_{\mathcal{M}^+})\|Q(\Vabc_{\mathcal{M}^+})) \nonumber \\
	&=& -I(\Vabc_M; \Vabc_{\text{pa}(M),1},\Vabc_{\text{pa}(M),2}) \nonumber \\
	&& - \sum_{i =1}^{M-1} I(\Vabc_i; \Vabc_{\text{pa}(i)}) + \text{constant}.
\end{eqnarray}
Therefore, to find the topology for weakly meshed system (e.g., maximum number of parents is less than two), the optimization problem in Lemma~\ref{thm:MI_sum} is approximated as
\begin{equation}
	\widehat{\boldsymbol{\Theta}} = \argmax_{\boldsymbol{\Theta} \subset \mathcal{M}^+} I(\Vabc_M; \Vabc_{\boldsymbol{\Theta}_{M,1}}, \Vabc_{\boldsymbol{\Theta}_{M,2}}) + \sum_{i=1}^{M-1} I(\Vabc_i; \Vabc_{\boldsymbol{\Theta}_i}).
\end{equation} 
Specifically, for each mutual information $I(\Vabc_M; \Vabc_{\boldsymbol{\Theta}_{M,1}}, \Vabc_{\boldsymbol{\Theta}_{M,2}})$, $I(\Vabc_M; \Vabc_{\boldsymbol{\Theta}_{M,1}}, \Vabc_{\boldsymbol{\Theta}_{M,2}}) + \sum_{i=1}^{M-1} I(\Vabc_i; \Vabc_{\boldsymbol{\Theta}_i})$ is computed by performing the maximum weighted spanning tree algorithm. Then, the one with the largest total mutual information is chosen to estimate system topology. The computational complexity is $\mathcal{O}(M(M-1)\log(M-1)$.

The method discussed above can be generalized to systems with more buses that contain more parents. \revv{However, the computational complexity also increases significantly. Therefore, for distributed systems that contain multiple loops, we recommend to adopt topology estimation methods that are designed for heavily mesh grids, such as \cite{liao2018urban,deka2017topology}.}

\subsection{Bus Phase Identification and Correction}
\revv{The previous section demonstrates that even with false phase labels, our method can correctly identify the multi-phase distribution grid topology.} In many field applications, accurate grid topology is not sufficient. The correct information of bus phases is also critical in grid plannings and operations. This subsection proposes a data-driven method to identify bus phase information and correct the false phase labels.
\begin{figure}[h!]
\vspace{-1ex}
\centering
\begin{circuitikz}
	\draw (0,0)  node[anchor=east] {$\Delta V^c_i$} 
	to[R=$z^{cc}$, i=$\Delta I^c$, o-o] (6,0) 
	node[anchor=west] {$\Delta V^c_k$};
	
	\draw (0,1) node[anchor=east] {$\Delta V^b_i$}
	to[R=$z^{bb}$, i=$\Delta I^b$, o-o] (6,1)
	node[anchor=west] {$\Delta V^b_k$};
	
	\draw (0,2) node[anchor=east] {$\Delta V^a_i$}
	to[R=$z^{aa}$, i=$\Delta I^a$, o-o] (6,2)
	node[anchor=west] {$\Delta V^a_k$};
	
	\draw (0,2.2) node[anchor=south] {Bus $i$};
	\draw (6,2.2) node[anchor=south] {Bus $k$};
	
	\draw (0.7,2) 
	to[R=$z^{ab}$, o-o] (0.7,1);
	
	\draw (4,2)
	to[short, o-] (4,1);
	
	\draw(4,1)
	to[R, l=$z^{ac}$, -o] (4,0);
	
	\draw (1.5,1)
	to[R=$z^{bc}$, o-o] (1.5,0);

\end{circuitikz}
\caption{An example of the two-port three-phase circuit.}
\label{fig:three_phase}
\vspace{-2ex}
\end{figure}
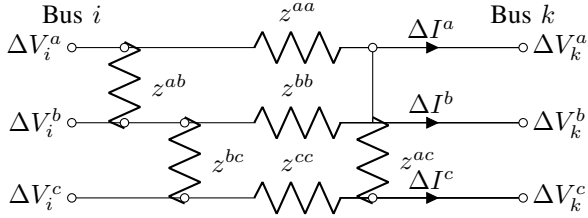

\begin{lemma}
\label{lemma:phase_ident}
In a multi-phase distribution grid, if two terminal buses of a branch are connected on the same phase, their phase voltage correlation is the largest.
\end{lemma}

\begin{proof}
Using the modified Carson's equation \cite{kersting2006distribution}, the self impedance $z^{aa}$ and mutual impedances $z^{ab}$ and $z^{ac}$ of a multi-phase power line can be computed as follows:
\begin{eqnarray}
z^{aa} &=& r_{ik} + 0.095 + j0.121\times H^a_{ik} \Omega\text{/miles}, \\
z^{ab} &=& 0.095 + j0.121\times H^{ab}_{ik} \Omega\text{/miles}, \\
z^{ac} &=& 0.095 + j0.121\times H^{ac}_{ik} \Omega\text{/miles},
\end{eqnarray}
where $H^a_{ik}$, $H^{ab}_{ik}$, and $H^{ac}_{ik}$ are constants, and $r^a_{ik}$ is resistance of branch $i-k$ in $\Omega\text{/miles}$. In a distribution grid, the resistance is usually larger than reactance \cite{baran1989network}. Therefore, $z^{aa} \simeq r^a_{ik} + 0.095$ and $z^{ab} \simeq z^{ac} \simeq 0.095$.

For bus $i$ and bus $k$, the voltages and currents can be expressed as
\begin{equation}
	\begin{bmatrix}
		\Delta V_i^a \\ \Delta V_i^b \\ \Delta V_i^c
	\end{bmatrix}
	= \begin{bmatrix}
		\Delta V_k^a \\ \Delta V_k^b \\ \Delta V_k^c
	\end{bmatrix}
	+
	\begin{bmatrix}
		Z^{aa} & Z^{ab} & Z^{ac} \\
		Z^{ab} & Z^{bb} & Z^{bc} \\
		Z^{ac} & Z^{bc} & Z^{cc} \\
	\end{bmatrix}
	\begin{bmatrix}
		\Delta I^a \\ \Delta I^b \\ \Delta I^c
	\end{bmatrix},
\end{equation}
where $Z^{mn} = z^{mn}\times l$ and $l$ is the line length. The equations above can be simplified to 
\begin{align}
	\Delta V_i^a &= \Delta V_k^a + C + r_{ik} \times l \times \Delta I^a, \\
	\Delta V_i^b &= \Delta V_k^b + C + r_{ik} \times l \times \Delta I^b, \\
	\Delta V_i^c &= \Delta V_k^c + C + r_{ik} \times l \times \Delta I^c,
\end{align}
where $C = 0.095\times l \times (\Delta I^a + \Delta I^b + \Delta I^c)$. The phase voltages at the two ends of a branch are in a linear relationship. Therefore, their correlation is the largest.
\end{proof}

Table.~\ref{tab:phase_corr} shows the voltage magnitude corrections among bus 64, 65, and 66 in IEEE 123-bus system. There are no PMUs in the system. Since PMUs can provide accurate phase measurements, the bus phase label identification problem is trivial with the presence of PMUs. In the 123-bus system, bus 64 and 65 are connected on phase $b$. Bus 65 and 66 are connected on phase $c$. In Table.~\ref{tab:phase_corr}, the correlation between bus 64 and 65 on phase $b$ is much larger than other pairs. Similar observation holds for bus 65 and 66. Thus, to identify bus phases in a distribution grid, the correlation check can be applied from the substation of the radial network down to all leaf buses. The reason is that the substation bus label information is usually reliable. \revv{Then, the bus phase can be correctly identified, following the paths of estimated grid topology. Note that, the metering device installed at each bus can provide the number of phases at each bus. Therefore, this method is eligible for all types of bus. The same approach can also be applied to diagnose the correctness of the bus phase labels.}

\begin{table}[htbp]
\caption{Voltage Magnitude Correlations Between Bus 64, 65, and 66.}
\label{tab:phase_corr}
\centering
\begin{tabular}{|c|c|c|c|}
\hline
& $|\Delta V_{65}^a|$ & $|\Delta V_{65}^b|$ & $|\Delta V_{65}^c|$ \\
\hline
$|\Delta V_{64}^b|$ & 0.4956 & 0.9996 & 0.5332 \\
\hline
$|\Delta V_{66}^c|$ & 0.9526 & 0.5479 & 1.0000 \\
\hline
\end{tabular}
\end{table}

\section{Unbalanced Multi-Phase Distribution Grid Topology Estimation with Incorrect Phase Labels}
\label{sec:alg_unbalance}
The results in the previous section illustrate the topology estimation for balanced multi-phase systems with incorrect phase labels. However, it is not directly expendable to unbalanced multi-phase systems. As shown in Fig.~\ref{fig:three_phase}, the voltages and currents are coupled cross different phases. Also, the unbalanced loads on each phase lead to the voltages angles are not separated by $2\pi/3$. To address these issues, the grid is transformed using sequence component frameworks. The voltage phasor is decomposed into three balanced phasors known as positive sequence, negative sequence, and zero sequence. The multi-phase voltage $\Vabc_i$ in phase frame is decomposed as follows:
\begin{equation}
\label{eq:phz_transform}
	\Vabc_i
	= 
	\begin{bmatrix}
		1 & 1 & 1 \\
		h^2 & h & 1\\
		h & h^2 & 1
	\end{bmatrix}
	\begin{bmatrix}
		\Delta V_a^p \\ \Delta V_a^n \\ \Delta V_a^z
	\end{bmatrix}
	=
	\mathbf{H}\Delta\mathbf{V}_a^{pnz},
\end{equation}
where $h = \exp{j2\pi/3}$, $h^2 = \exp{-j2\pi/3}$, $\Delta V_a^p, \Delta V_a^n, \Delta V_a^z$ denote positive-sequence, negative-sequence, and zero-sequence voltage on phase $a$. $\Vpnz_a$ is called the sequence voltage of phase $a$. Since each sequence component system is balanced, the sequence component voltages of phase $b$ and phase $c$ are the phase shifts of voltage on phase $a$, $\Vpnz_a$. Thus, the sequence components voltages of phase $b$ and $c$ are not required to compute. In the following text, $\Vpnz_i$ denotes the sequence voltage vector of bus $i$ on phase $a$. The sequence voltages can be computed as follows:
\begin{equation}
\label{eq:abc_transform}
	\Vpnz = \mathbf{H}^{-1}\Vabc = \frac{1}{3}\mathbf{H}^H\Vabc,
\end{equation}
where the operator $H$ denotes the Hermitian transpose. The same transformation can also be applied to the multi-phase current phasors and admittance matrix, i.e.,
\begin{eqnarray}
	\Ipnz_i &=& \mathbf{H}^{-1}\Iabc_i, \\
	\Ypnz_{ik} &=& \mathbf{H}^{-1}\Yabc_{ik}\mathbf{H}.
\end{eqnarray}
A highlight is that the transformation above is applied to the multi-phase voltage phasors, current phasors, and admittance matrix at a particular bus, not the entire system. Therefore, if two buses are not connected, e.g., $\Yabc_{ik} = \Bzero$, $\Ypnz_{ik} = \Bzero$. Therefore, finding topology in phase frame is equivalent to finding topology in sequence component frame. The proof of $P(\Vpnz_{\mathcal{M}^+})=\prod_{i=1}^M P(\Vpnz_i|\Vpnz_{\text{pa}(i)})$ is required to apply the mutual information-based maximum weight spanning tree algorithm (Chow-Liu algorithm).

The transformation process in (\ref{eq:abc_transform}) does not require the correct phase labels in the phase frame. The reason is that when the bus phase labels are incorrect, the decomposition in (\ref{eq:abc_transform}) will become either the sequence component frame $\Vpnz_b$ or $\Vpnz_c$. Since both $\Vpnz_b$ or $\Vpnz_c$ are both balanced systems, the same method proposed for $\Vpnz_a$ can be applied to estimate system topology.

\begin{lemma}[Data Processing Inequality \cite{cover2012elements}]
\label{thm:data_processing}
	If random vectors $\mbf{X}, \mbf{Y}, \mbf{Z}$ forms a Markov Chain, i.e., $\mbf{X} \rightarrow \mbf{Y} \rightarrow \mbf{Z}$, $I(\mbf{X}; \mbf{Y}) \geq I(\mbf{X}; \mbf{Z})$. Also, for the function of $\mbf{Y}$, $g(\mbf{Y})$, $I(\mbf{X}; \mbf{Y}) \geq I(\mbf{X}; g(\mbf{Y}))$.
\end{lemma}

\begin{lemma}
\label{thm:cond_indept_phase}
	Consider a multi-phase distribution grid and assume that the current injection increment at each bus is approximately independent, e.g., $\Iabc_i \perp \Iabc_k$ for $i\neq k$. Given the nodal bus voltage increment of bus $i$ in sequence component frame, the nodal bus voltage increments of every pair of bus $i$'s children are conditionally independent, i.e., $\Vpnz_k \perp \Vpnz_l | \Vpnz_i$ for $ k,l \in \mathcal{C}(i)$ and $k\neq l$.
\end{lemma}

\begin{proof}
The first step of the proof is showing that the current injection increment are independent in sequence component frame, given the current injection increment at each bus is independent in phase frame. There are multiple ways to prove it. Here, an information theoretical approach is adopted.

When two random vectors are independent, their mutual information is zero \cite{cover2012elements}, e.g., $I(\Iabc_i; \Iabc_k) = 0$ if $\Iabc_i \perp \Iabc_k$. Since $\Ipnz_k = \mathbf{H}^{-1}\Iabc_k$ is a linear transformation of $\Iabc_k$, these random vectors form a Markov Chain, i.e., $\Iabc_i \rightarrow \Iabc_k \rightarrow \Ipnz_k$. Applying Lemma~\ref{thm:data_processing},
\begin{equation}
I(\Iabc_i; \Ipnz_k) \leq I(\Iabc_i; \Iabc_k) = 0.
\end{equation}
Because the mutual information is non-negative, $I(\Iabc_i; \Ipnz_k)=0$. $\Iabc_i = \mathbf{H}\Ipnz_i$ is a function of $\Ipnz_i$. Thus, another Markov Chain is formed: $\Ipnz_i \rightarrow \Iabc_i \rightarrow \Ipnz_k$. Applying Lemma~\ref{thm:data_processing} again, 
\begin{equation}
I(\Ipnz_i; \Ipnz_k) \leq I(\Iabc_i; \Ipnz_k) = 0.
\end{equation}

$I(\Ipnz_i; \Ipnz_k)$ is zero due to the non-negativity of mutual information. Therefore, if the current injections are independent in phase frame, they are also independent in sequence component frame.

The second step of the proof is showing that the conditional independence of nodal voltages holds in sequence component frame. The example in Fig.~\ref{fig:6bus_example} is adopted to illustrate it. In the sequence component frame, the nodal equation of the system in Fig.~\ref{fig:6bus_example} is $\Ypnz_{\mathcal{M}^+}\Vpnz_{\mathcal{M}^+} = \Ipnz_{\mathcal{M}^+}$, where $\Ypnz_{\mathcal{M}^+}$ is
\begin{equation}
\begin{bmatrix}
    \Ypnz_{11} & \Ypnz_{12} & \Ypnz_{13} & \Bzero & \Bzero & \Bzero & \Bzero\\
    \Ypnz_{21} & \Ypnz_{22} & \Bzero & \Ypnz_{24} & \Ypnz_{25} & \Bzero & \Bzero\\
    \Ypnz_{31} & \Bzero & \Ypnz_{33} & \Bzero & \Bzero & \Ypnz_{36} & \Ypnz_{37}\\
    \Bzero & \Ypnz_{42} &  \Bzero & \Ypnz_{44} & \Bzero & \Bzero & \Bzero\\
    \Bzero & \Ypnz_{52} &  \Bzero & \Bzero & \Ypnz_{55} & \Bzero & \Bzero\\
    \Bzero & \Bzero &  \Ypnz_{63} & \Bzero & \Bzero & \Ypnz_{66} & \Bzero \\
    \Bzero & \Bzero &  \Ypnz_{73} & \Bzero & \Bzero & \Bzero & \Ypnz_{77}
\end{bmatrix},
\end{equation}
$\Ypnz_{ik} = \Ypnz_{ki}$, and $\Ypnz_{ii} =  -\sum_{k=1,k\neq i}^7 \Ypnz_{ik}$. If $\Ypnz_{ik} = \Bzero$, there is no branch between bus $i$ and $k$. This equation is in the same format as (\ref{eq:nodal_eqn}). Since $\Ipnz_i \perp \Ipnz_k$ for all $i \neq k$, the same method used in the proof of Lemma~\ref{thm:cond_indept} and Lemma~\ref{lemma:one_hop_indept} can show the conditional independence of nodal voltages in sequence component frame.
\end{proof}

With Lemma~\ref{thm:cond_indept_phase}, the conditional independence of current injection is proved to hold in the sequence component frame as well, e.g., $P(\Vpnz_{\mathcal{M}^+}) = \prod_{i=1}^M P(\Vpnz_i|\Vpnz_{\text{pa}(i)})$. Since the sequence component system is a balanced multi-phase system, the Chow-Liu algorithm can estimate topology in the sequence component frame.

\begin{theorem}
\label{thm:chow_liu_seq}
	In an unbalanced radial distribution grid, the topology can be estimated by solving the following problem
	\begin{equation}
	\widehat{\boldsymbol{\Theta}} = \argmax_{\boldsymbol{\Theta} \subset \mathcal{M}^+} \sum_{i=1}^M I(\Vpnz_i; \Vpnz_{\boldsymbol{\Theta}_i}).
	\end{equation}
	Also, the mutual information-based maximum weight spanning tree algorithm (Chow-Liu algorithm) solves the problem above.	
\end{theorem}

The proof of Theorem~\ref{thm:chow_liu_seq} is omitted here because it is similar to the proofs of Lemma~\ref{thm:MI_sum} and Theorem~\ref{thm:chow-liu}. The topology estimation algorithm for unbalanced multi-phase distribution grids is summarized in Algorithm~\ref{alg:topology}.
\begin{algorithm}[htbp]
\caption{Unbalanced Multiphase Distribution Grid Topology Estimation via Sequence Component Frame}
\label{alg:topology}
\begin{algorithmic}[1]
\REQUIRE $\vabc_i[n]$ for $i \in \mathcal{M}^+$, $n = 1, \cdots N$
\STATE Compute voltage phasor $\vpnz_i[n]$ using (\ref{eq:abc_transform}) for $i \in \mathcal{M}^+$, $n = 1, \cdots N$.
\FOR{$i,k \in \mathcal{M}^+$}
\STATE Compute empirical mutual information $I(\Vpnz_i;\Vpnz_k)$ based on $\vpnz_i[n]$ and $\vpnz_k[n]$.
\ENDFOR
\STATE Sort all possible bus pair $(i,k)$ into non-increasing order by $I(\Vpnz_i,\Vpnz_k)$. Let $\mathcal{T}$ denote the sorted set.
\STATE Repeat Step 5 to Step 16 in Algorithm~\ref{alg:topology00}.
\end{algorithmic}
\end{algorithm}

The phase angles of $\Vabc$ are needed for performing the phase frame transformation in (\ref{eq:abc_transform}). However, as discussed in Section~\ref{sec:vmag}, PMUs have not been widely available in distribution grids. To only use voltage magnitudes to address the unbalance problem, the following approximation is proposed:
\begin{equation}
\label{eq:vpnz_mag}
\Vpnz = \mathbf{H}^{-1}|\Vabc|.
\end{equation}
In this approximation, only the voltage magnitudes in phase frame are used to compute voltages in sequence component frame. As demonstrated in Section~\ref{sec:sim}, this approximation does not introduce significant errors to topology estimation. In addition, the grid topology is identical in phase frame and in sequence component frame. Therefore, once the unbalanced grid topology is estimated, Lemma~\ref{lemma:phase_ident} is applied to identify the phases of all buses.

\section{Simulations and Numerical Results}
\label{sec:sim}
In this section, the proposed algorithms for balanced and unbalanced grid topology estimation are validated on on \revv{IEEE $37$-bus, $123$-bus (Fig.~\ref{fig:123bus}), and $8500$-bus distribution networks \cite{kersting2001radial, dugan2010ieee} using data from USA (California and Taxes) and Europe. Also, we validate the proposed algorithms on systems with different levels of DER penetration and the presence of incorrect phase labels. Furthermore, sensitivity analysis is conducted on data lengths, data accuracy, load patterns, and data resolutions.}

\revv{$37$-bus, $123$-bus, and $8500$-bus systems are multi-phase.} In each network, the feeder or substation is selected as the slack bus (bus $0$). The historical data have been preprocessed by the GridLAB-D \cite{chassin2008gridlab}, an open source simulator for distribution grid. The load profile from PG\&E is used to simulate the power system behavior in a practical pattern. This profile contains anonymized and secure hourly smart meter readings over $110,000$ PG\&E residential customers for a period of one year spanning from $2011$ to $2012$. Since both $37$-bus and $123$-bus systems are primary distribution grids, the real power at each bus is an aggregation of $10 - 100$ customers. The load buses in both systems are unnecessary to be multi-phase. The details of bus and branch phases are given in \cite{kersting2001radial}. The voltage data at each bus are used for topology estimation. Fig.~\ref{fig:flowchart2} summarizes the overall process of topology estimation and bus phase identification.

\putFig{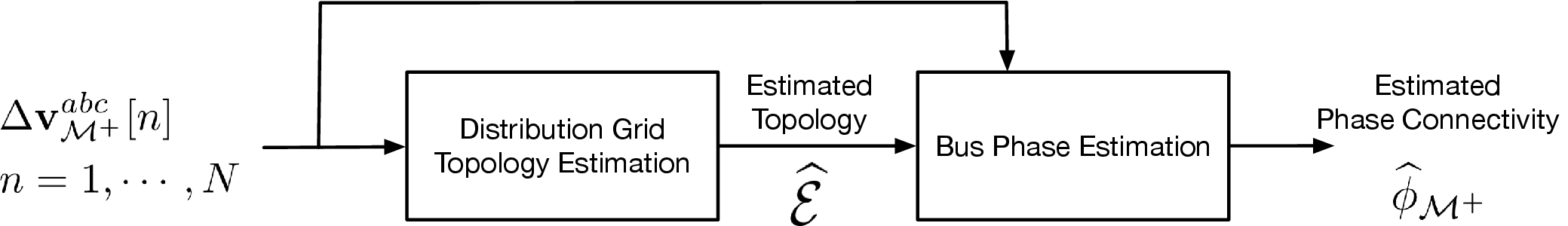}{Flow chart of topology estimation and bus phase identification process.}{1\linewidth}

\putFig{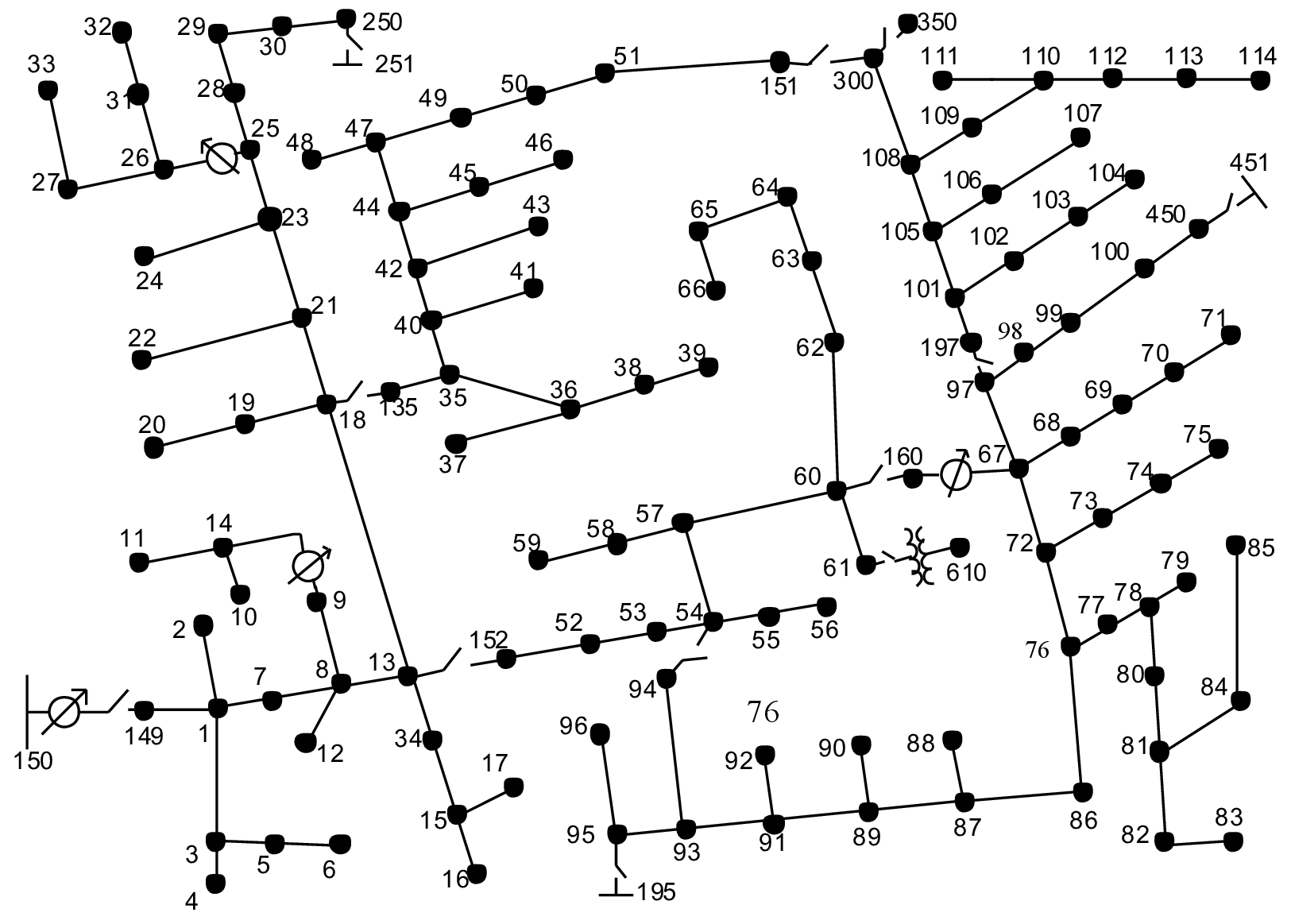}{IEEE $123$-bus distribution test case.}{0.9\linewidth}


The PG\&E data set does not contain the reactive power. The reactive power $q_i^\phi[n]$ on phase $\phi$ of bus $i$ at time $n$ is computed according to a random lagging power factor $pf^\phi_i[n]$, which follows a uniform distribution, e.g., $pf^\phi_i[n] \sim \Unif(0.8,0.95)$. To obtain voltage time-series, i.e., $\mathbf{v}_i[n]$, the power flow analysis is run to generate the hourly states of the power system over a year. $N = 8760$ measurements are obtained at each bus. Section~\ref{sec:length} investigates the data length requirement for topology estimation. The loads attached to each phase are unequal. Hence, the systems are unbalanced. Fig.~\ref{fig:P_37bus2} and Fig.~\ref{fig:P_123bus2} show the hourly aggregated real powers on each phase in 37-bus and 123-bus systems. Although each phase has the similar pattern over time, the magnitudes of real powers are different on each phase. Therefore, the testing systems are unbalanced. 

\putFig{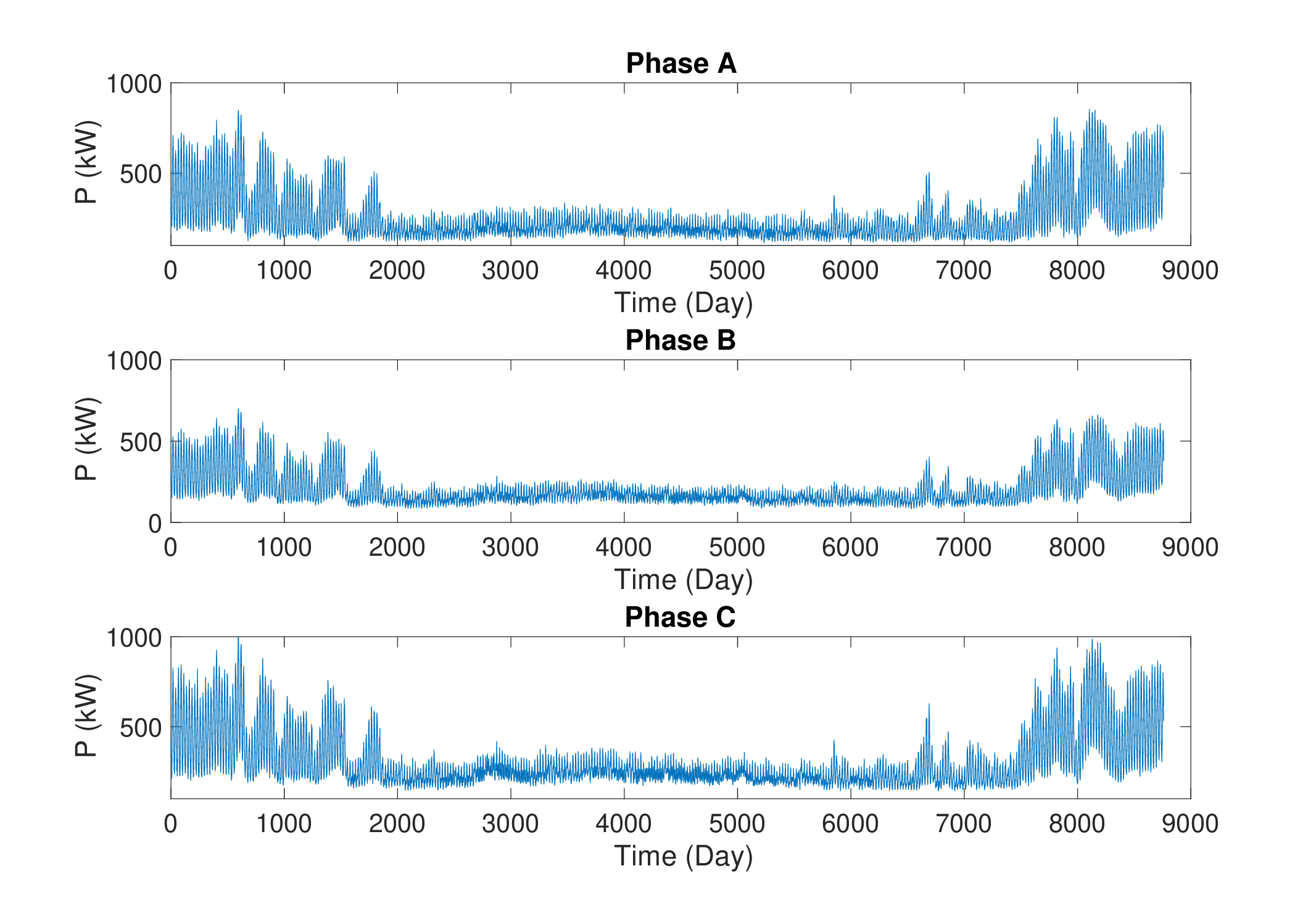}{Hourly aggregated real powers on each phase in 37-bus system.}{1.1\linewidth}
\putFig{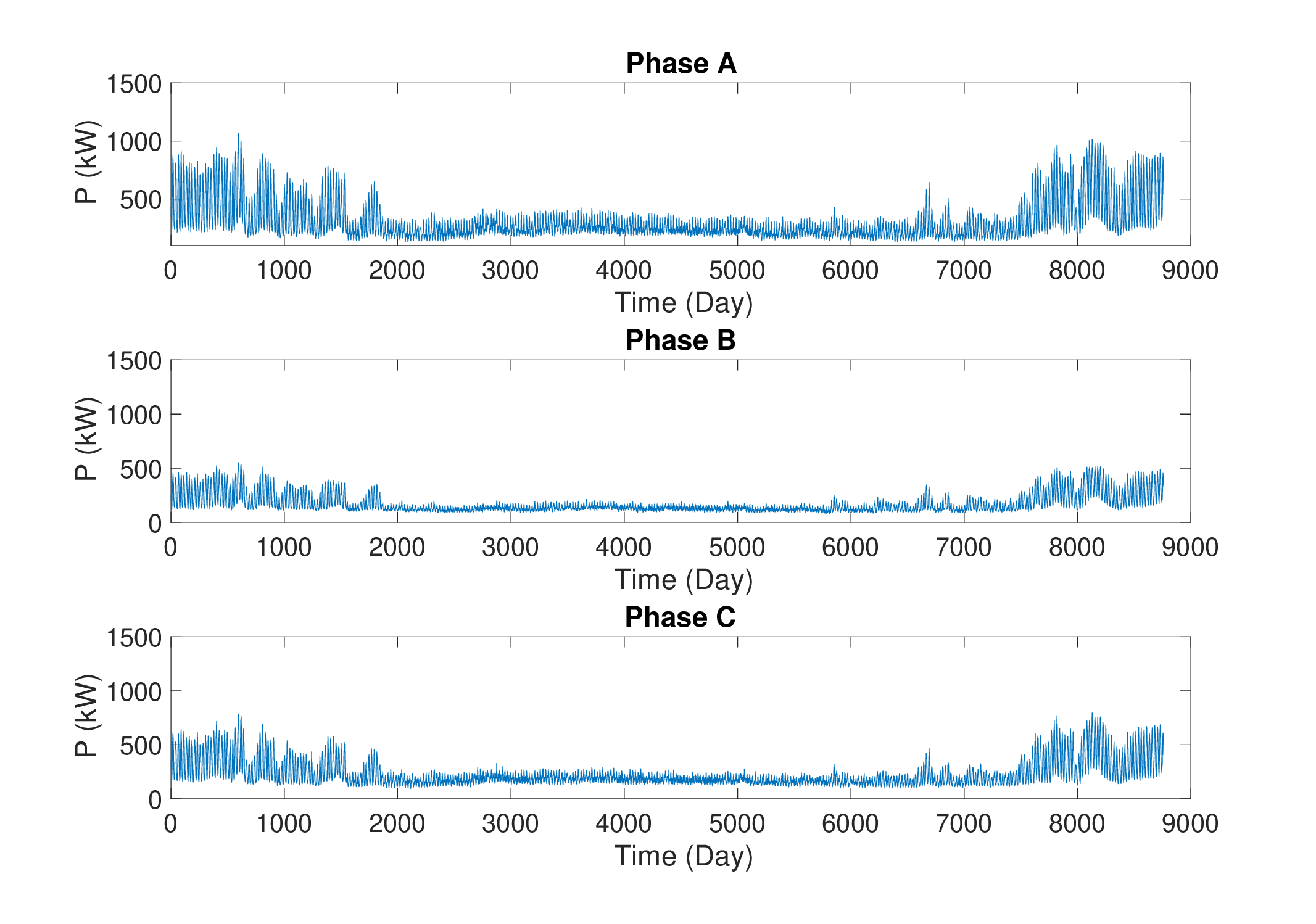}{Hourly aggregated real powers on each phase in 123-bus system.}{1.1\linewidth}

\subsection{Distribution Grid Topology Estimation Error Rate}
This section discusses the performance on grid topology estimation. The error rate (ER) is employed as the performance evaluation metric, which is defined as
\begin{equation}
\text{ER} = \frac{1}{|\mathcal{E}|}\bigg(\underbrace{\sum_{(i,k) \in \widehat{\mathcal{E}}} \indic{(i,k) \notin \mathcal{E}}}_{\text{false estimation}} + \underbrace{\sum_{(i,k) \in \mathcal{E}} \indic{(i,k) \notin \widehat{\mathcal{E}}}}_{\text{missing}}\bigg)\% 
\end{equation}
where $\widehat{\mathcal{E}}$ denotes the edge set estimates, $|\mathcal{E}|$ is the size of $\mathcal{E}$, and $\indic{.}$ is the indicator function. The first and second terms represent the number of falsely estimated branches and the number of missing branches, respectively.


Table~\ref{tab:mismatch} summarizes the topology estimation error rates of unbalanced multi-phase 37- and 123-bus systems using noiseless data. When phase angle data are available, our algorithm perfectly estimates the grid topology. When only voltage magnitudes are available, our algorithm can still estimate the grid topology perfectly. This result also verifies that our approximation in (\ref{eq:vpnz_mag}) is sufficient for topology estimation.
\begin{table}[htbp]
\caption{Topology Estimation Error Rate without DERs.}
\label{tab:mismatch}	
\centering
\begin{tabular}{|c||c|c||c|c|}
	\hline
	& \multicolumn{2}{c||}{Proposed Method} & \multicolumn{2}{c|}{Modified Single-Phase} \\
	& \multicolumn{2}{c||}{} & \multicolumn{2}{c|}{Method} \\
	\hline
	System & $\Vpnz$ & $|\Vpnz|$ & $\Vpnz$ & $|\Vpnz|$ \\
	\hline
	37-bus & $0.00\%$ & $0.00\%$ & $5.56\%$ & $8.33\%$ \\
	\hline
	123-bus & $0.00\%$ & $0.00\%$ & $1.64\%$ & $1.64\%$ \\
	\hline
\end{tabular}
\end{table}

The proposed algorithm is also compared with a modified single-phase topology estimation in \cite{weng2017distributed}. Specifically, the single-phase topology estimator is applied to each phase individually. Then, the single-phase topology estimates are combined to produce the multi-phase system topology. As shown in Table~\ref{tab:mismatch}, the modified single-phase method has worse performance than the proposed algorithm. The key reason is that the modified single-phase method does not consider the voltage coupling across phases.

The proposed algorithm is compared with the method in \cite{zhao2018learning}, which is also based on minimizing the KL distance and searches the correct operational topology from all possible topology candidates. For 37-bus system, the error rate of \cite{zhao2018learning} is $5.6\%$. For 123-bus system, the error rate is $8.2\%$. The high error rates are due to the DC approximation in \cite{zhao2018learning}. For unbalanced distribution grids, the DC approximation does not hold generally.

In addition, we validate our algorithm on IEEE 8500-node distribution system \cite{dugan2010ieee}, which contains both low-voltage and medium-voltage buses. The error rate that using $\Vpnz$ is 15.7\%. Most incorrect identified branches are near the low-voltage grid feeders. In many systems, the locations and connectives of the low-voltage grid feeders are accurate. Therefore, with the prior knowledge of low-voltage grid feeders, the error rate is reduced to 3.8\%.

\begin{table}[htbp]
	\caption{Topology Estimation Computation Time (seconds)}
	\label{tab:compute_time}	
	\centering
	\begin{tabular}{|c|c|c|c|}
	\hline
	System & \makecell{Computation time of \\ mutual information} & \makecell{Computation time of \\ maximum weight \\ spanning tree} & \makecell{Total \\ time} \\
	\hline
	$37$-bus & $0.476$ & $0.499$ & $0.975$ \\
	\hline
	$123$-bus & $2.978$ & $3.114$ & $6.092$ \\
	\hline 
	$8500$-bus & $244.345$ & $478.111$ & $722.456$ \\
	\hline 
	\makecell{$8500$-bus \\ (parallel)} & $89.090$ & $94.028$ & $183.118$ \\
	\hline
	\end{tabular}
\end{table}
\revv{
Table~\ref{tab:compute_time} summarizes the average computational time of the proposed algorithm on different systems over 1000 Monte Carlo simulation iterations. For $37$-bus and $123$-bus systems, our algorithm takes a few seconds to report the estimated topology, which makes it suitable for real-time monitoring. For $8500$-bus system, the computational time of both mutual information and maximum weight spanning tree grows up. Though the computational time for a large system is high, some power system properties can help to speed up the topology estimation process. As mentioned above, the locations and connectives of low-voltage grid feeders are accurate. Hence, for large-scale system that has both low-voltage and medium-voltage systems, the topology estimation problem can be performed in two steps: 1) only identify the topology of low-voltage grids and 2) only estimate the topology of medium-voltage grids. Since each low-voltage grid operates independently, the topology estimation process can run in parallel. As indicated in Table~\ref{tab:compute_time}, by decomposing a large-scale grid into multiple small sub-grids, the computational time of topology estimation is reduced by $75\%$. A highlight is that in the parallel computation, the maximum computational time is bounded by the largest low-voltage grid. If every low-voltage grid is small (e.g., similar size as the $123$-bus system), the computational time can be much less. Another highlight is that the computational time is invariant to the integration of DERs and data lengths.
}

\subsection{Distribution Grids with DER Integration}
The penetration of DERs has grown significantly during last decade and will keep increasing in the future. As discussed earlier, the high penetration of DER will lead to a deeply unbalanced distribution grid. To evaluate the proposed algorithm with integrated DERs, $20\%$ of residents in the distribution networks are selected to install rooftop photovoltaic (PV) systems. The profiles of hourly power generation are obtained from NREL PVWatts Calculator, an online simulator that estimates the PV power generation based on weather history of PG\&E service zone and the physical parameters of a $5$kW PV panel in residential levels \cite{dobos2014pvwatts}. The power factor is fixed as $0.90$ lagging, which satisfies the regulation of many U.S. utilities \cite{ellis2012review} and IEEE standard \cite{ieee2014guide}. \revv{Similar to the simulations without DERs, we use one year's data ($8760$ samples) to estimate topology.}

\begin{table}[htbp]
\caption{Topology Estimation Error Rate with $20\%$ PV Penetrations.}
\label{tab:mismatch_der}	
\centering
\begin{tabular}{|c||c|c||c|c|}
	\hline
	& \multicolumn{2}{c||}{Proposed Method} & \multicolumn{2}{c|}{Modified Single-Phase} \\
	& \multicolumn{2}{c||}{} & \multicolumn{2}{c|}{Method} \\
	\hline
	System & $\Vpnz$ & $|\Vpnz|$ & $\Vpnz$ & $|\Vpnz|$ \\
	\hline
	37-bus & $0.00\%$ & $0.00\%$ & $8.33\%$ & $11.11\%$ \\
	\hline
	123-bus & $0.00\%$ & $0.00\%$ & $1.64\%$ & $1.64\%$ \\
	\hline
\end{tabular}
\end{table}

The error rates of grid topology estimation with the rooftop PVs integration are presented in Table~\ref{tab:mismatch_der} using noiseless measurements. Our algorithm does not have any performance degradation with DER integration. Also, the modified single-phase method still performs worse than the proposed method. Compared with the systems without DER, the modified single-phase method has performance degradation.

In order to further validate the proposed algorithm, the DER penetration level is progressively increased from $0\%$ to $100\%$. For each penetration level, Monte Carlo simulation is performed over $1000$ iterations. Fig.~\ref{fig:diff_level_DER} plots the error rate with different levels of DER penetration using the voltage magnitude $|\Vpnz|$ only. Besides $60\%$ penetration of DERs, the error rates do not change with the growth of DER installation rate, which highlights the reliability of the proposed algorithm. 12 iterations of Monte Carlo simulation have errors when the DER penetration level is $60\%$. The incorrect identified branches are bus $57$ - bus $58$ and bus $58$ - bus $59$. The loads with PV integrations on these three buses are similar and the line impedances are identical. This causes that the voltage profiles of these buses are similar. Our algorithm is hard to identify the correct connectivity. However, such an instance requires the same impedance and same voltage profiles. This rarely happens in practice. As the penetration level increase, this instance is not observed again and the proposed algorithm can correctly identify these two branches.

\putFig{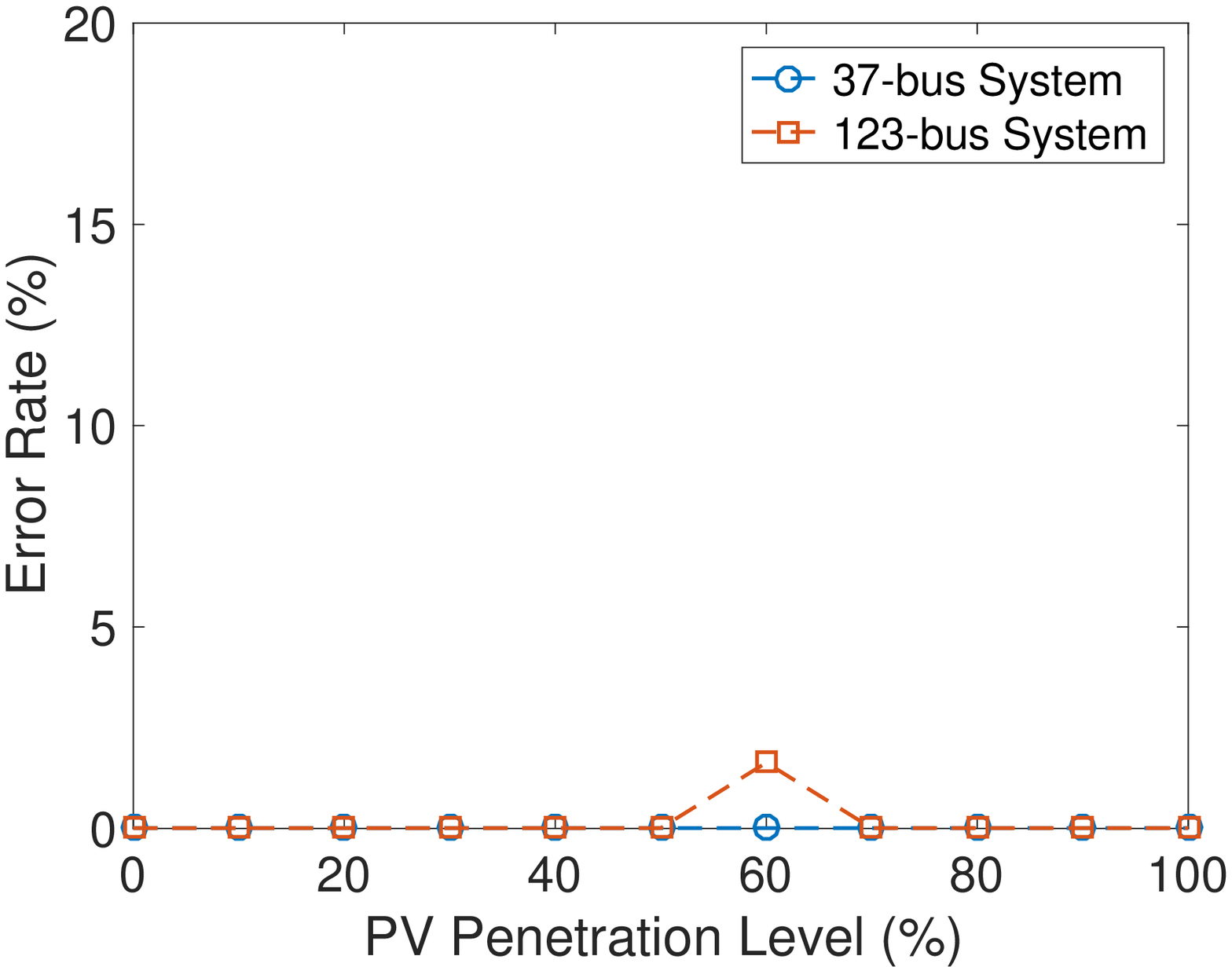}{Error rates with different levels of DER penetration using $|\Vpnz|$.}{1\linewidth}
\subsection{Distribution Grids with Incorrect Phase Labels}
\label{sec:incorrect_phase}
In some distribution grids, up to $10\%$ of the phase labels are incorrect or unknown. Therefore, this section validates our algorithm on the 123-bus system with incorrect phase labels. To simulate the incorrect phase labels, several buses are randomly chosen and switch their phase $a$ voltage measurements to data of either phase $b$ voltage or phase $c$ voltage.

Table~\ref{tab:incorrect_label} shows the error rates with different percentages of incorrect phase labels using voltage magnitude only and highlights that our algorithm is insensitive to incorrect bus phase label. As discussed previously, the incorrect phase labels is a permutation of random variables in $|\mathbf{V}_i^{pnz}|$ and do not affect $I(|\mathbf{V}_i^{pnz}|;|\mathbf{V}_k^{pnz}|)$. If the modified single-phase approach is used, the error rate increases significantly because the mutual information is computed for incorrect bus pairs. For the 123-bus system, the error rate is $11.7\%$ when $10\%$ buses have incorrect phase labels.

\begin{table}[htbp]
	\caption{Error Rate with Incorrect Phase Labels using $|\Vpnz|$.}
	\centering
	\begin{tabular}{|c|c|c|}
	\hline
	Percentage of Bus with & Error Rate & Error Rate \\
	Incorrect Phase Labels & Average & Standard Deviation \\
	\hline
	2\% & 0\% & 0\% \\ \hline
	6\% & 0\% & 0\% \\ \hline
	10\% & 0\% & 0\% \\ \hline
	14\% & 0\% & 0\% \\ \hline
	18\% & 0\% & 0\% \\	\hline
	20\% & 0\% & 0\% \\
	\hline
	\end{tabular}
	\label{tab:incorrect_label}
\end{table}

\subsection{Sensitivity Analysis}
\subsubsection{Sensitivity to Data Lengths}
\label{sec:length}
The proposed algorithm is validated with different data lengths, ranging from $1$ to $360$ days. Fig.~\ref{fig:data_len} illustrates the error rates of the 123-bus system, with and without DER, over different lengths of the PG\&E data set. With $20$ days' measurements ($24\times 20 = 480$ data points), the proposed method can achieve zero error. This result is better than the single-phase system presented in \cite{weng2017distributed}, which requires 30 days' observations. The reason is that at time $n$, our proposed algorithm uses measurements from three phases, which contain more information than the single-phase system. The frequency of distribution grid reconfiguration ranges from hours to weeks \cite{jabr2014minimum,dorostkar2016value}. Section~\ref{sec:resolution} demonstrates that the topology can still be estimated by increasing the sampling frequency of smart meters.

\putFig{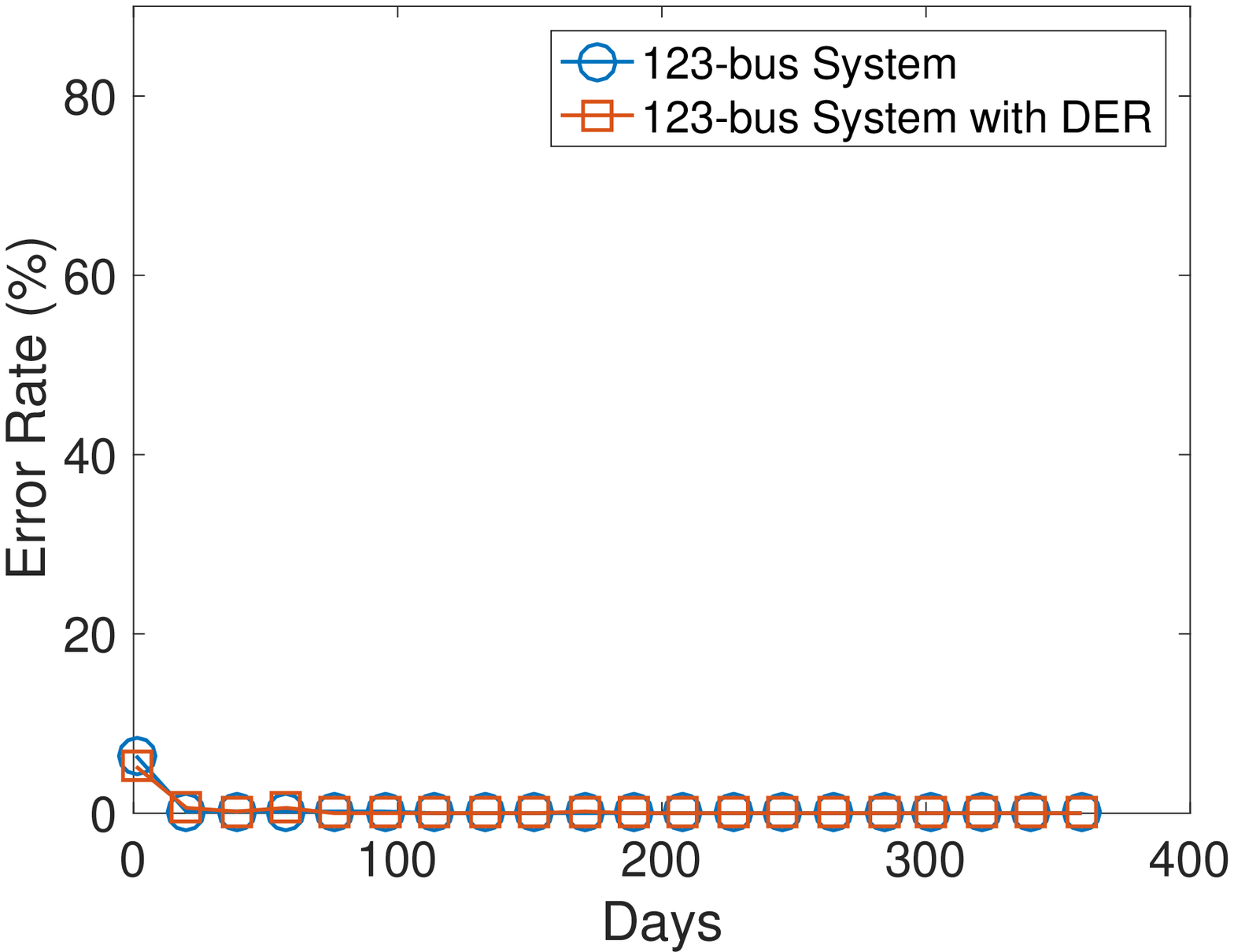}{Error rates with different data lengths.}{1\linewidth}

\subsubsection{Sensitivity to Data Accuracy}
In particles, smart meter measurements are noisy. Thus, it is important to validate our algorithm under different levels of measurement noises. In the U.S., ANSI C12.20 standard (Class 0.5) requires the smart meters to have an error less than $\pm 0.5\%$ \cite{zheng2013smart,ansc12}. Table.~\ref{tab:noise_error} shows the error rates with different noise levels over $20$ iterations in the $123$-bus system with PG\&E data. Compared with the estimation results using perfect measurements, the error rates grow up as the increase in noise levels. These newly introduced errors are around the feeders. For example, bus $251$ and $451$ are both feeders and incorrectly connected with the presence of noise. In real systems, the location of feeder buses are usually known. Therefore, a post-processing can be applied on the topology estimate and remove these unnecessary branches from topology estimate. The updated system is still a radial network. After performing post-processing, the error rate decreases to $1\%$. In Table.~\ref{tab:noise_error}, the standard deviation of error rate is very small and therefore, our algorithm can provide reliable and consistent results with noisy measurements.

\begin{table}[h!]
\caption{Error Rates with Different Voltage Noise Levels in $123$-bus System}
\label{tab:noise_error}
\centering
\begin{tabular}{|c|c|c|}
	\hline
	Noise Level & Error Rate & Error Rate \\
	& Average & Standard Deviation \\
	\hline
	0.01\% & 2.95\% & 0.41\% \\
	\hline 
	0.05\% & 2.91\% & 0.50\% \\
	\hline
	0.1\% & 2.99\% & 0.67\% \\
	\hline
	0.2\% & 3.98\% & 0.77\% \\
	\hline
\end{tabular}
\end{table}

\subsubsection{Sensitivity to Data Patterns}
The ``ADRES-Concept'' Project load profile \cite{Einfalt11,VUT16} is employed to understand our algorithm's sensitivity to load patterns. This data set contains real and reactive power profiles of 30 houses in Upper-Austria. The data were sampled every second over 7 days in summer and 7 days in winter. The voltage data are generated using the 37-bus system. The load profiles are scaled to match the scale of power in the 37-bus system. The resulting multi-phase system is unbalanced.

Fig.~\ref{fig:EU_pattern} compares the error rates using summer and winter load profiles. When there is only one measurement, the proposed algorithm has $200\%$ error rate due to poor estimation of mutual information. The error rate is above $100\%$ because all estimated branches are incorrect and none of the correct branch is found. As more measurements become available, the error rate reduces significantly. Also, our algorithm has a consistent performance in winter and summer. Compared with the results in \cite{weng2017distributed}, the proposed algorithm perfectly estimates the grid topology with shorten time because more information is observed at each time step.

\putFig{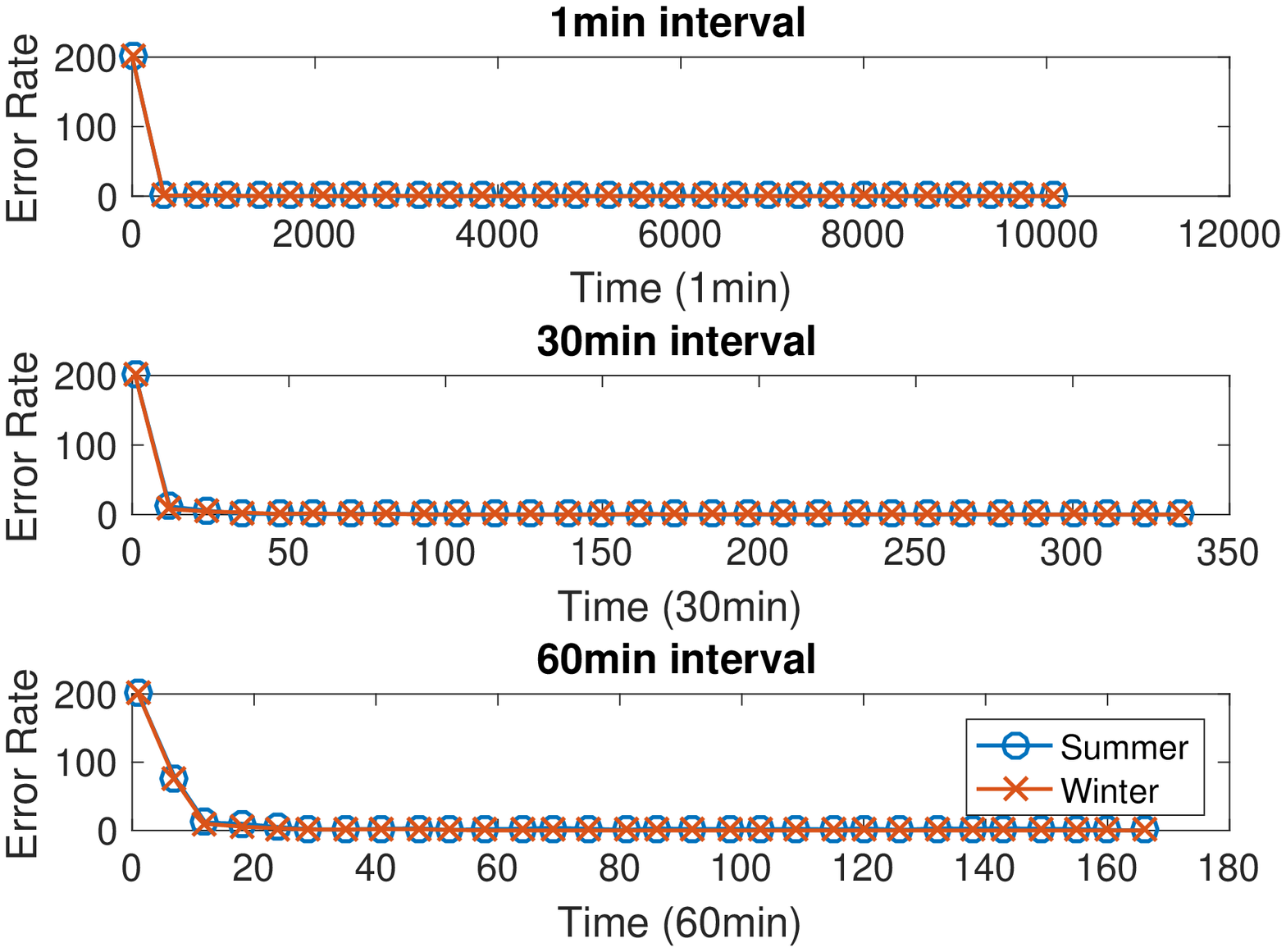}{Error rates with summer and winter load patterns.}{\linewidth}

Another validation of our algorithm is using data set from Pecan Street, which contains hourly load measurements of $345$ houses with PV integrations in Austin, Taxes. The measurements include both power consumption and renewable generation. In Fig.~\ref{fig:pecan_st_diff}, our algorithm requires $16$ hours' measurements to recover the topology of the 37-bus system, which is similar to the ADRES data set. This highlights the robustness of our algorithm.
\putFig{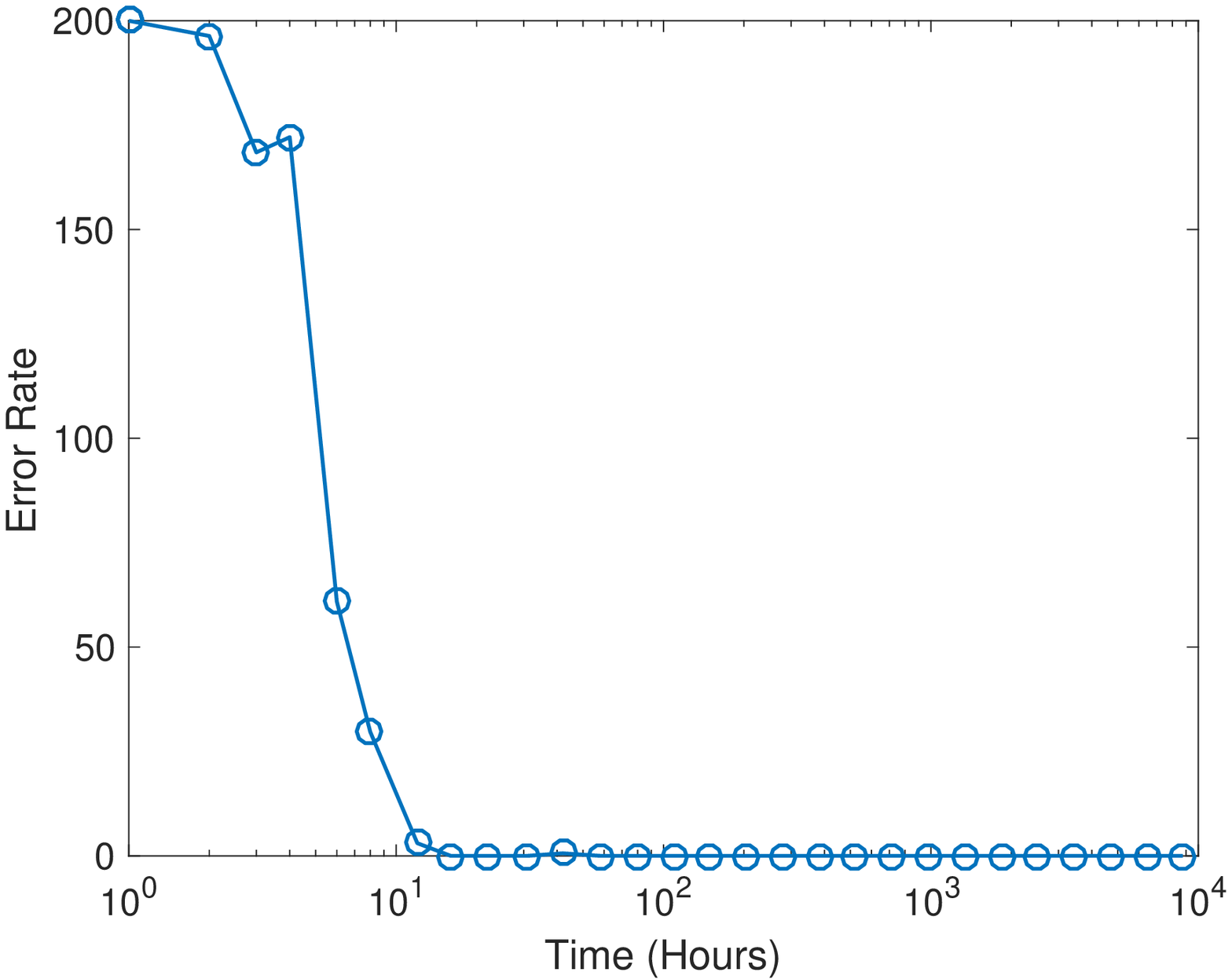}{Error rates on Pecan Street data set.}{\linewidth}

\revv{
In order to better understanding the impacts of ZIP loads and high applicants on the topology estimation, the applicant/device simulation model \cite{chassin2008gridlabd} provided by Gridlab-D is adopted to generate load data. This simulator is based on the thermal data and device configurations. Therefore, compared with the real data provided from PG\&E, a detailed setup of load patterns is possible. For each residential load, multiple devices and applicants (e.g., heating, electric hot water heaters, washer and dryers, cooking, electronic plugs and lights) are installed with various configurations and parameters. We run the simulation on IEEE $123$-bus network with real temperature data from Palo Alto, CA. The simulation is performed on an hourly basis for one year's duration. By applying the proposed algorithm, the topology can still be correctly estimated and the required data lengths are consistent with the results in Section~\ref{sec:length}.
}

\subsubsection{Sensitivity to Data Resolutions}
\label{sec:resolution}
Fig.~\ref{fig:EU_resolution2} illustrates the performance of the proposed algorithm under different sampling frequencies using the ADRES data set. When the sampling period is $1$ minute, about $6$ hours' voltage profile are required to perfectly recover the system. According to \cite{dorostkar2016value}, some distribution grids reconfigure as fast as every $3$ hours. Therefore, the proposed algorithm is suitable for existing systems and real-time operations. If the sampling period is $30$ minutes, $35$ data points ($35 \times 30 \text{ minutes} = 17.5 \text{hours}$) to recover the system topology. This estimation time is only half of the required time in \cite{weng2017distributed}.

\putFig{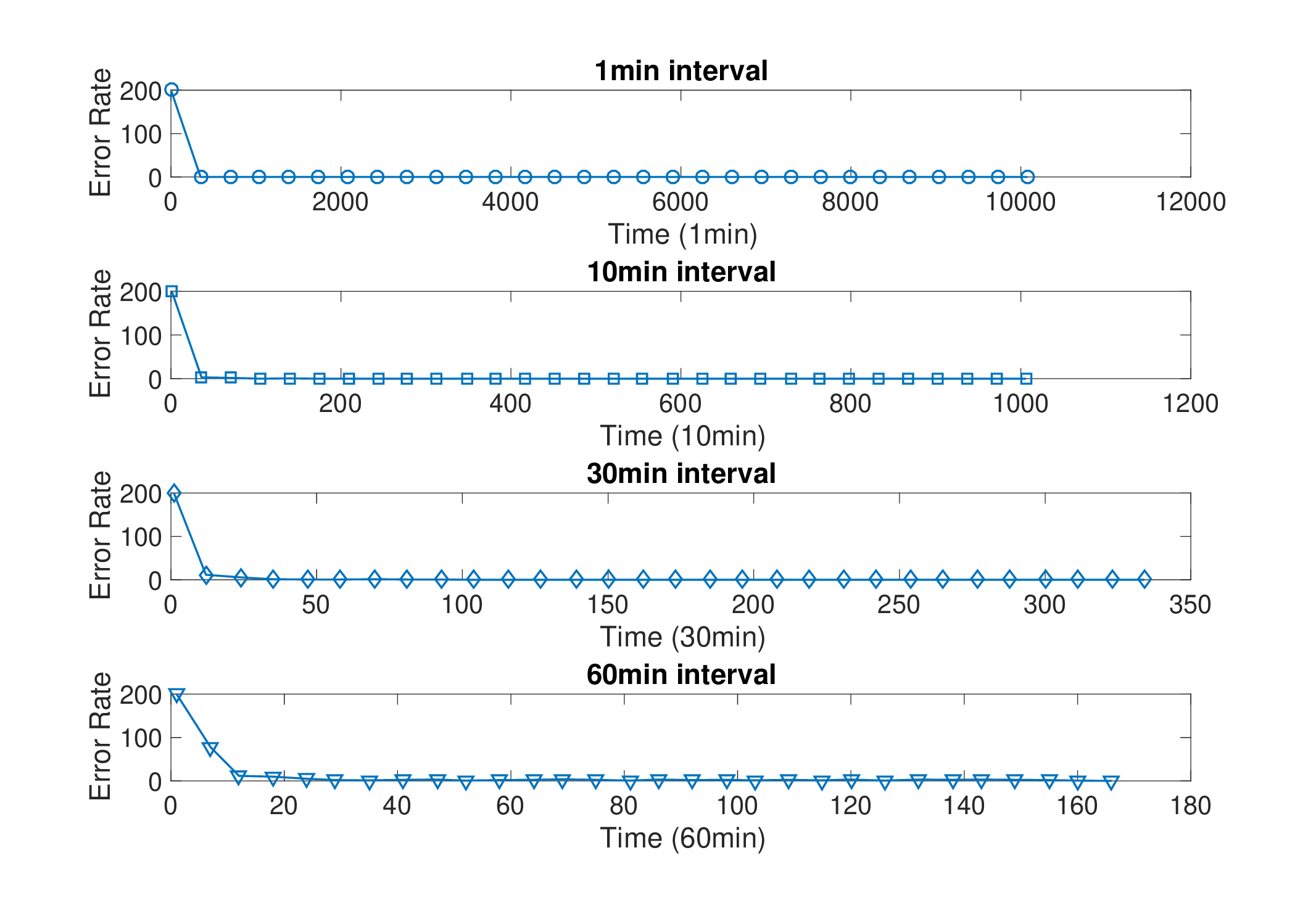}{Error rates with different data resolutions.}{1.1\linewidth}

\section{Conclusions}
\label{sec:con}
This paper proposes a data-driven approach to estimate multi-phase distribution grid topology by utilizing smart meter measurements. Unlike existing approaches, our method does not require the system to be balanced. Also, our method tolerates the errors of bus phase labels. Specifically, the topology estimation problem is formulated as a joint distribution (voltage phasors) approximation problem under the probabilistic graphical model framework. Then, the distribution grid topology estimation is proven to be equivalent to the graphical model estimation problem and propose a mutual information-based maximum weight spanning tree algorithm, which is optimal and efficient. Moreover, our algorithm is extended to the case where only voltage magnitude is available. In addition, as bus phase labels are critical to distribution grid plannings and operations, a simple approach is introduced to correct the error of bus phase labels by utilizing Carson's equations. Finally, the proposed algorithm is validated on IEEE $37$- and $123$-bus systems and compared with the existing single-phase method. Results show that the proposed algorithm outperforms the single-phase method and has robust performances when bus phase labels are incorrect. Our algorithms are also validated under different penetration levels of DERs and conduct the sensitivity analysis. The numerical results are highly accurate and robust in various system configurations.

\section{Appendices}
\subsection{Proof of Lemma~\ref{thm:cond_indept}}
\label{sec:proof_cond_indept}
\begin{proof}
\putFig{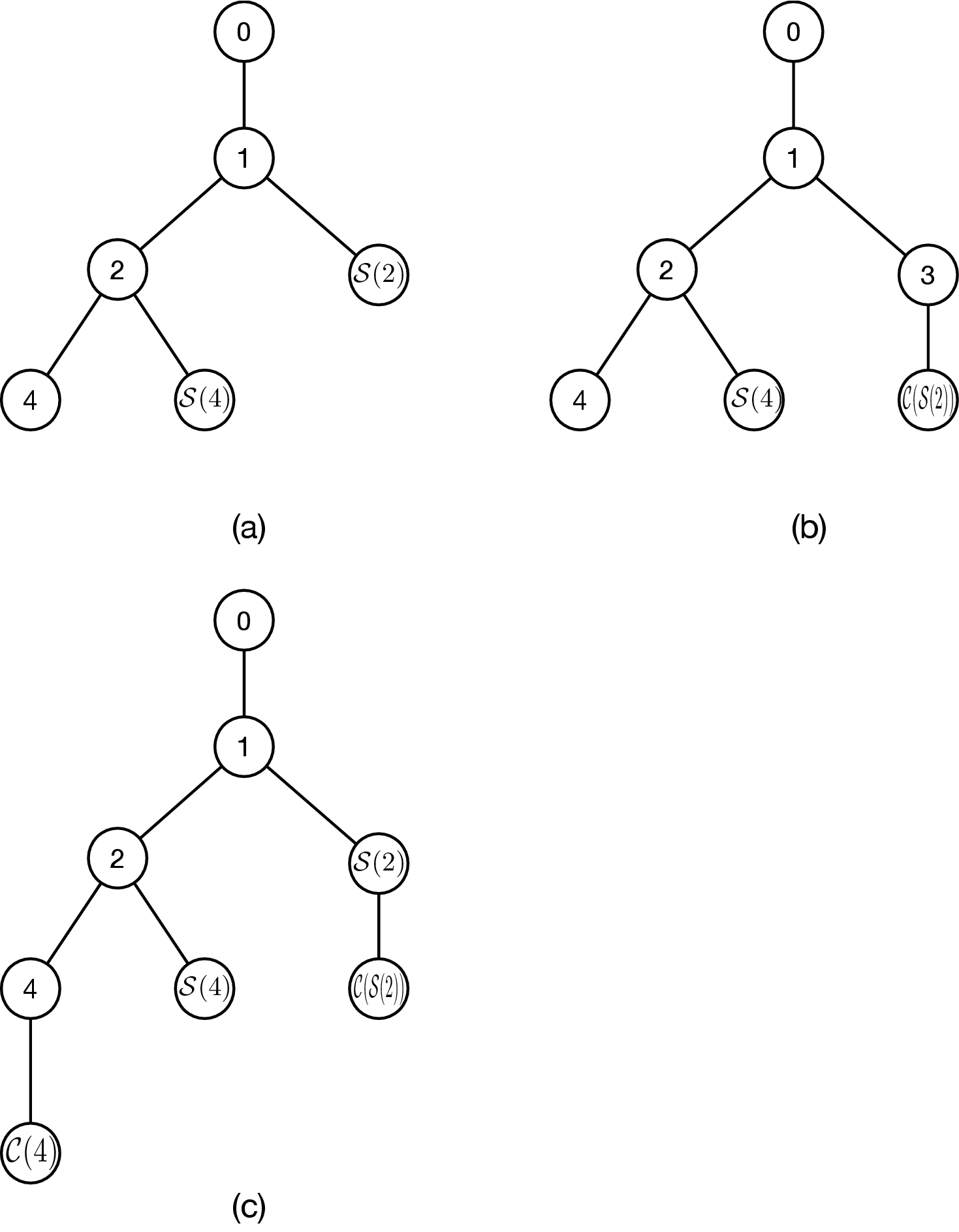}{Figure for the proof of Lemma~\ref{thm:cond_indept}.}{0.9\linewidth}
Several cases illustrated in Fig.~\ref{fig:indept_proof} are used to prove Lemma~\ref{thm:cond_indept}. The first step is proving the leaf nodes. In Fig.~\ref{fig:indept_proof}(a), for bus $4$, given $\Vabc_2 = \vabc_2$, $\Vabc_{\mathcal{S}(4)} = \vabc_{\mathcal{S}(4)}$, and $\Vabc_1 = \vabc_1$:
\begin{eqnarray}
	\Iabc_4 &=& \Yabc_{42}\vabc_2 + \Yabc_{44}\Vabc_4, \\
	\Iabc_k &=& \Yabc_{1k}\vabc_1 + \Yabc_{kk}\Vabc_k \quad \forall k \in \mathcal{S}(2).
\end{eqnarray}
Since $\Iabc_4 \perp \Iabc_k$, $\Vabc_4$ and $\Vabc_k$ are conditionally independent for $k \in \mathcal{S}(2)$. This results can be generalized to all leaf buses that share with same grandparent bus ($\Vabc_1$).

In Fig.~\ref{fig:indept_proof}(b), for bus $4$, given $\Vabc_2 = \vabc_2$, $\Vabc_{\mathcal{S}(4)} = \vabc_{\mathcal{S}(4)}$, and $\Vabc_1 = \vabc_1$:
\begin{eqnarray}
	\Iabc_4 &=& \Yabc_{42}\vabc_2 + \Yabc_{44}\Vabc_4, \\
	\Iabc_3 &=& \Yabc_{13}\vabc_1 + \sum_{k\in \mathcal{C}(3)}\Yabc_{3k}\Vabc_k + \Yabc_{33}\Vabc_3, \label{eq:proof_eq22} \\
	\Iabc_k &=& \Yabc_{3k}\Vabc_3 + \Yabc_{kk}\Vabc_k \quad \forall k \in \mathcal{C}(3). \label{eq:proof_eq2}
\end{eqnarray}

Since $\Yabc_{ii} = -\Yabc_{\pa{i}i} - \sum_{k \in \mathcal{C}(i)}\Yabc_{ki}$ and $\Yabc_{ik} = \Yabc_{ki}$, combining (\ref{eq:proof_eq22}) and (\ref{eq:proof_eq2}), the equation becomes
\begin{eqnarray}
&&\Iabc_3 + \sum_{k \in \mathcal{C}(3)} \Iabc_k \nonumber \\
&=& \Yabc_{13}\vabc_1 + \Yabc_{33}\Vabc_3 + \sum_{k\in \mathcal{C}(3)}\Yabc_{3k}\Vabc_k \nonumber \\
&& +\sum_{k\in \mathcal{C}(3)}(\Yabc_{3k}\Vabc_3 + \Yabc_{kk}\Vabc_k) \nonumber \\
&=& \Yabc_{13}\vabc_1 - (\Yabc_{13}\Vabc_3 + \sum_{k \in \mathcal{C}(3)}\Yabc_{3k}\Vabc_3) \nonumber \\
&& + \sum_{k \in \mathcal{C}(3)}(\Yabc_{3k}\Vabc_k + \Yabc_{3k}\Vabc_3 - \Yabc_{3k}\Vabc_k) \nonumber \\
&=& \Yabc_{13}\vabc_1 - \Yabc_{13}\Vabc_3 \label{eq:proof_eq1}
\end{eqnarray}
Given $\Iabc_4 \perp (\Iabc_3 + \sum_{k\in \mathcal{C}(3)} \Iabc_k)$, $\Vabc_4$ and $\Vabc_3$ are conditionally independent. (\ref{eq:proof_eq1}) can be rewritten as an equation of $\Vabc_3$, i.e.,
\begin{equation}
\Vabc_3 = (\Yabc_{13})^{-1}(\Yabc_{13}\vabc_1 - \Iabc_3 -\sum_{k\in \mathcal{C}(3)} \Iabc_k).
\end{equation}
Replacing $\Vabc_3$ in (\ref{eq:proof_eq2}) with the equations above, then, for $k \in \mathcal{C}(3)$,
\begin{equation}
\Yabc_{kk}\Vabc_k+\Yabc_{3k}\vabc_1 = \Iabc_k + \Yabc_{3k}(\Yabc_{13})^{-1}(\Iabc_3 + \sum_{k\in \mathcal{C}(3)} \Iabc_k).
\end{equation}
Given $\Iabc_4$ and $\Iabc_k + \Yabc_{3k}(\Yabc_{13})^{-1}(\Iabc_3 + \sum_{k\in \mathcal{C}(3)} \Iabc_k)$ are independent and $\vabc_1$ is a constant, $\Vabc_4$ and $\Vabc_k$ are conditionally independent for $k \in \mathcal{C}(3)$. When there are more child buses $\Vabc_{\mathcal{C}(3)}$, the same induction method above can be applied to prove the conditional independence. Thus, the proof of Fig.~\ref{fig:indept_proof}(b) can be generalized to prove the conditional independence of a leaf bus and all other buses that are under the same grandparent bus.

Next part proves the lemma for non-leaf buses. In Fig.~\ref{fig:indept_proof}(c), for bus $4$, given $\Vabc_2 = \vabc_2$, $\Vabc_{\mathcal{S}(4)} = \vabc_{\mathcal{S}(4)}$, and $\Vabc_1 = \vabc_1$,
\begin{eqnarray}
	\Iabc_4 &=& \Yabc_{42}\vabc_2 + \Yabc_{44}\Vabc_4 + \sum_{k \in \mathcal{C}(4)}\Yabc_{4k}\Vabc_k, \label{eq:proof_eq3} \\
	\Iabc_k &=& \Yabc_{4k}\Vabc_4 + \Yabc_{kk}\Vabc_k, \forall k \in \mathcal{C}(4), \label{eq:proof_eq4} \\
	\Iabc_l &=& \Yabc_{1l}\vabc_1 + \sum_{m \in \mathcal{C}(l)}\Yabc_{lm}\Vabc_m + \Yabc_{ll}\Vabc_l,\label{eq:proof_eq5} \\
	\Iabc_m &=& \Yabc_{lm}\Vabc_l + \Yabc_{mm}\Vabc_m, \label{eq:proof_eq6}
\end{eqnarray}
where $l \in \mathcal{S}(2)$ and $m\in \mathcal{C}(l)$. Combining (\ref{eq:proof_eq3}) and (\ref{eq:proof_eq4}) yields
\begin{equation}
\label{eq:proof_eq7}
	\Iabc_4 + \sum_{k \in \mathcal{C}(4)} \Iabc_k = \Yabc_{42}\vabc_2 - \Yabc_{42}\Vabc_4.
\end{equation}
For every $l$ in $\mathcal{S}(2)$, combining (\ref{eq:proof_eq5}) and (\ref{eq:proof_eq6}) yields the following equation:
\begin{equation}
\label{eq:proof_eq8}
	\Iabc_l + \sum_{m \in \mathcal{C}(l)} \Iabc_m = \Yabc_{1l}\vabc_1 - \Yabc_{1l}\Vabc_l.	
\end{equation}
Applying the strategy in Fig.~\ref{fig:indept_proof}(b) to (\ref{eq:proof_eq7}) and (\ref{eq:proof_eq8}) could prove that $\Vabc_4$ and $\Vabc_l$ are conditionally independent. Also, $\Vabc_4$ and $\Vabc_{m}$ are proved to be conditionally independent for $m \in \mathcal{C}(l)$ by combining (\ref{eq:proof_eq6}) and (\ref{eq:proof_eq8}). The results in Fig.~\ref{fig:indept_proof}(c) can be generalized to all non-leaf buses. Using the results in Fig.~\ref{fig:indept_proof}, Lemma~\ref{thm:cond_indept} is proved to hold.
\end{proof}

\subsection{Proof of Theorem~\ref{thm:MI_sum}}
\label{app:proof_MI_sum}
\begin{proof}
Recall the definition \cite{cover2012elements}, the KL divergence is expressed as
\begin{align}
& D(P(\Vabc_{\mathcal{M}^+})\|Q(\Vabc_{\mathcal{M}^+}; \boldsymbol{\Theta})) \nonumber \\
=& E_{P(\Vabc_{\mathcal{M}^+})}\log\frac{P(\Vabc_{\mathcal{M}^+})}{Q(\Vabc_{\mathcal{M}^+}; \boldsymbol{\Theta})} \nonumber \\
=&\int P(\Vabc_{\mathcal{M}^+}) \log \frac{P(\Vabc_{\mathcal{M}^+})}{Q(\Vabc_{\mathcal{M}^+}; \boldsymbol{\Theta})} \nonumber \\
=& \int P(\Vabc_{\mathcal{M}^+}) \left(\log P(\Vabc_{\mathcal{M}^+}) - \log Q(\Vabc_{\mathcal{M}^+};\boldsymbol{\Theta})\right) \nonumber \\
=& \int P(\Vabc_{\mathcal{M}^+}) \log P(\Vabc_{\mathcal{M}^+}) \nonumber \\
& - \int P(\Vabc_{\mathcal{M}^+})\log Q(\Vabc_{\mathcal{M}^+};\boldsymbol{\Theta}).
\end{align}

Because of Lemma~\ref{lemma:one_hop_indept}, the radial structured PDF $Q(\Vabc_{\mathcal{M}^+}; \boldsymbol{\Theta})$ can be expressed as a conditional distribution $\prod_{i = 1}^M P(\Vabc_i|\Vabc_{\boldsymbol{\Theta}_i})$. Then,
\begin{align}
& D(P(\Vabc_{\mathcal{M}^+})\|Q(\Vabc_{\mathcal{M}^+}; \boldsymbol{\Theta})) \nonumber \\
=& \int P(\Vabc_{\mathcal{M}^+}) \log P(\Vabc_{\mathcal{M}^+}) \nonumber \\
& - \int P(\Vabc_{\mathcal{M}^+})\log \prod_{i = 1}^M P(\Vabc_i|\Vabc_{\boldsymbol{\Theta}_i}) \nonumber \\
=& \int P(\Vabc_{\mathcal{M}^+}) \log P(\Vabc_{\mathcal{M}^+}) \nonumber \\
& - \int P(\Vabc_{\mathcal{M}^+})\sum_{i = 1}^M \log P(\Vabc_i|\Vabc_{\boldsymbol{\Theta}_i}),
\end{align}
where $P(\Vabc_1|\Vabc_0) = P(\Vabc_1) $ due to the fact that $\Vabc_0$ is a constant. By following the definition of conditional probability and adding $P(\Vabc_i)$ into the denominator, onecan simplify the equation above as
\begin{align}
& D(P(\Vabc_{\mathcal{M}^+})\|Q(\Vabc_{\mathcal{M}^+}; \boldsymbol{\Theta})) \nonumber \\
=& \int P(\Vabc_{\mathcal{M}^+}) \log P(\Vabc_{\mathcal{M}^+}) \nonumber \\
& - \int P(\Vabc_i|\Vabc_{\boldsymbol{\Theta}_i})\sum_{i=1}^M\log\frac{P(\Vabc_i, \Vabc_{\boldsymbol{\Theta}_i})}{P(\Vabc_i)P(\Vabc_{\boldsymbol{\Theta}_i})} \nonumber \\
& - \sum_{i=1}^M \int P(\Vabc_i) \log P(\Vabc_i) \nonumber \\
=& - H(\Vabc_{\mathcal{M}^+}) -\sum_{i=1}^M I\left(\Vabc_i; \Vabc_{\boldsymbol{\Theta}_i}\right) + \sum_{i=1}^M H(\Vabc_i) .
\end{align}
The last equality is due to the definitions of entropy, i.e.,
\begin{equation}
\label{eq:def_entropy}	
H(\Vabc_i) = -\int P(\Vabc_i)\log P(\Vabc_i),
\end{equation}
and mutual information, i.e.,
\begin{eqnarray}
	&& I\left(\Vabc_i; \Vabc_{\boldsymbol{\Theta}_i}\right) \nonumber \\
	&=& \int P(\Vabc_i,\Vabc_{\boldsymbol{\Theta}_i})\log\frac{P(\Vabc_i,\Vabc_{\boldsymbol{\Theta}_i})}{P(\Vabc_i)P(\Vabc_{\boldsymbol{\Theta}_i})}.\label{eq:def_mutual_info}
\end{eqnarray}

Thus, to minimize the KL-divergence between $P(\Vabc_{\mathcal{M}^+})$ and $Q(\Vabc_{\mathcal{M}^+};\boldsymbol{\Theta})$, one can choose the $M-1$ edges to maximize $\sum_{i=1}^M I\left(P(\Vabc_i); P(\Vabc_{\boldsymbol{\Theta}_i})\right)$. The entropy term $\sum_{i=1}^M H(\Vabc_i) - H(\Vabc_{\mathcal{M}^+})$ is irrelevant with the topology structure of distribution grid and is excluded in the final optimization problem. Therefore, minimizing $D(P(\Vabc_{\mathcal{M}^+})\|Q(\Vabc_{\mathcal{M}^+}; \boldsymbol{\Theta}))$ is equivalent to solving the following optimization problem:
\begin{equation}
\widehat{\boldsymbol{\Theta}} = \argmax_{\boldsymbol{\Theta} \subset \mathcal{M}^+} \sum_{i=1}^M I\left(\Vabc_i; \Vabc_{\boldsymbol{\Theta}_i}\right).
\end{equation}
\end{proof}

\bibliographystyle{iet}
\bibliography{ref}

\begin{thebibliography}{10}

\bibitem{dey2010urban}
Dey, S., Jessa, A., Gelbien, L.
\newblock `{Urban Grid Monitoring Renewables Integration}'.
\newblock In: IEEE Conference on Innovative Technologies for an Efficient and
  Reliable Electricity Supply. (,  2010. pp.~ 252--256

\bibitem{clement2010impact}
Clement.Nyns, K., Haesen, E., Driesen, J.: `{The Impact of Charging Plug-in
  Hybrid Electric Vehicles on a Residential Distribution Grid}', \emph{Power
  Systems, IEEE Transactions on},  2010, \textbf{25}, (1), pp.~371--380

\bibitem{huang2012electric}
Huang, J., Gupta, V., Huang, Y.F.
\newblock `{Electric Grid State Estimators for Distribution Systems with
  Microgrids}'.
\newblock In: IEEE 46th Annul Conference on Information Sciences and Systems.
  (,  2012. pp.~ 1--6

\bibitem{abur2004power}
Abur, A., Exposito, A.G.: `{Power System State Estimation: Theory and
  Implementation}'.
\newblock (CRC press,  2004)

\bibitem{Lugtu80}
Lugtu, R.L., Hackett, D.F., Liu, K.C., Might, D.D.: `{Power System State
  Estimation: Detection of Topological Errors}', \emph{Power Apparatus and
  Systems, IEEE Transactions on},  1980, \textbf{PAS-99}, (6), pp.~2406--2412

\bibitem{rudin2012machine}
Rudin, C., Waltz, D., Anderson, R.N., Boulanger, A., Salleb.Aouissi, A., Chow,
  M., et~al.: `Machine {Learning} for the {New} {York} {City} {Power} {Grid}',
  \emph{Pattern Analysis and Machine Intelligence, IEEE Transactions on},
  2012, \textbf{34}, (2), pp.~328--345

\bibitem{deka2015structure}
Deka, D., Chertkov, M., Backhaus, S.: `{Structure Learning in Power
  Distribution Networks}', \emph{IEEE Transactions on Control of Network
  Systems},  2017, \textbf{PP}, (99), pp.~1--1

\bibitem{cavraro2017voltage}
Cavraro, G., Kekatos, V., Veeramachaneni, S.: `{Voltage Analytics for Power
  Distribution Network Topology Verification}', \emph{IEEE Transactions on
  Smart Grid},  2017,

\bibitem{sharon2012topology}
Sharon, Y., Annaswamy, A.M., Motto, A.L., Chakraborty, A.
\newblock `{Topology Identification in Distribution Network with Limited
  Measurements}'.
\newblock In: Innovative Smart Grid Technologies (ISGT), 2012 IEEE PES. (IEEE,
  2012. pp.~ 1--6

\bibitem{bolognani2013identification}
Bolognani, S., Bof, N., Michelotti, D., Muraro, R., Schenato, L.
\newblock `{Identification of Power Distribution Network Topology via Voltage
  Correlation Analysis}'.
\newblock In: Conference on Decision and Control. (,  2013. pp.~ 1659--1664

\bibitem{peppanen2016distribution}
Peppanen, J., Grijalva, S., Reno, M.J., Broderick, R.J.
\newblock `{Distribution System Low-Voltage Circuit Topology Estimation using
  Smart Metering Data}'.
\newblock In: Transmission and Distribution Conference and Exposition. (IEEE,
  2016. pp.~ 1--5

\bibitem{liao2018urban}
Liao, Y., Weng, Y., Liu, G., Rajagopal, R.: `{Urban MV and LV Distribution Grid
  Topology Estimation via Group Lasso}', \emph{IEEE Transactions on Power
  Systems},  2018, pp.~ 1--1

\bibitem{lo1993decomposed}
Lo, K.L., Zhang, C.
\newblock `{Decomposed Three-Phase Power Flow Solution using the Sequence
  Component Frame}'.
\newblock In: IEE Proceedings C (Generation, Transmission and Distribution).
  vol. 140. (IET,  1993. pp.~ 181--188

\bibitem{tleis2007power}
Tleis, N.: `{Power Systems Modelling and Fault Analysis: Theory and Practice}'.
\newblock (Newnes,  2007)

\bibitem{ansc95}
{American National Standards Institute}.
\newblock `{ANSI C84.1: Electric power systems and equipment voltage ratings
  (60 Herz)}'.
\newblock ({American National Standards Institute},  1995.

\bibitem{routtenberg2015pmu}
Routtenberg, T., Xie, Y., Willett, R.M., Tong, L.: `{PMU-based Detection of
  Imbalance in Three-Phase Power Systems}', \emph{IEEE Transactions on Power
  Systems},  2015, \textbf{30}, (4), pp.~1966--1976

\bibitem{yuan2016inverse}
Yuan, Y., Ardakanian, O., Low, S., Tomlin, C.: `On the inverse power flow
  problem', \emph{arXiv preprint arXiv:161006631},  2016,

\bibitem{ardakanian2017event}
Ardakanian, O., Yuan, Y., Dobbe, R., von Meier, A., Low, S., Tomlin, C.
\newblock `{Event Detection and Localization in Distribution Grids with Phasor
  Measurement Units}'.
\newblock In: 2017 IEEE Power \& Energy Society General Meeting. (IEEE,  2017.
  pp.~ 1--5

\bibitem{deka2018topology}
Deka, D., Chertkov, M., Backhaus, S.: `{Topology Estimation using Graphical
  Models in Multi-Phase Power Distribution Grids}', \emph{arXiv preprint
  arXiv:180306531},  2018,

\bibitem{weng2017distributed}
Weng, Y., Liao, Y., Rajagopal, R.: `{Distributed Energy Resources Topology
  Identification via Graphical Modeling}', \emph{IEEE Transactions on Power
  Systems},  2017, \textbf{32}, (4), pp.~2682--2694

\bibitem{kersting2006distribution}
Kersting, W.H.: `Distribution system modeling and analysis'.
\newblock (CRC press,  2006)

\bibitem{kersting2001radial}
Kersting, W.H.
\newblock `{Radial Distribution Test Feeders}'.
\newblock In: IEEE Power Engineering Society Winter Meeting. vol.~2. (,  2001.
  pp.~ 908--912

\bibitem{Einfalt11}
Einfalt, A., Schuster, A., Leitinger, C., Tiefgraber, D., Litzlbauer, M.,
  Ghaemi, S., et~al.: `{ADRES-Concept: Konzeptentwicklung f{\"u}r
  ADRES-Autonome Dezentrale Regenerative EnergieSysteme}', \emph{TU Wien,
  Institut f{\"u}r Elektrische Anlagen und Energiewirtschaft},  2011,

\bibitem{VUT16}
{Institute of Energy Systems and Electrical Drives}.
\newblock `{ADRES-Dataset}'.
\newblock (,  2016.
\newblock Available from:
  \url{http://www.ea.tuwien.ac.at/projects/adres\_concept/EN/}

\bibitem{chassin2008gridlab}
Chassin, D.P., Schneider, K., Gerkensmeyer, C.
\newblock `{GridLAB-D: An Open-Source Power Systems Modeling and Simulation
  Environment}'.
\newblock In: Transmission and distribution conference and exposition, 2008.
  t\&d. IEEE/PES. (IEEE,  2008. pp.~ 1--5

\bibitem{laughton1968analysis}
Laughton, M.
\newblock `{Analysis of Unbalanced Polyphase Networks by the Method of Phase
  Co-ordinates. Part 1: System Representation in Phase Frame of Reference}'.
\newblock In: Proceedings of the Institution of Electrical Engineers. vol. 115.
  (IET,  1968. pp.~ 1163--1172

\bibitem{chen1991distribution}
Chen, T.H., Chen, M.S., Hwang, K.J., Kotas, P., Chebli, E.A.: `{Distribution
  System Power Flow Analysis-A Rigid Approach}', \emph{IEEE Transactions on
  Power Delivery},  1991, \textbf{6}, (3), pp.~1146--1152

\bibitem{kersting2012distribution}
Kersting, W.H.: `{Distribution System Modeling and Analysis}'.
\newblock (CRC press,  2012)

\bibitem{liao2016urbanpes}
Liao, Y., Weng, Y., Rajagopal, R.
\newblock `Urban distribution grid topology reconstruction via lasso'.
\newblock In: 2016 IEEE Power and Energy Society General Meeting (PESGM). (,
  2016. pp.~ 1--5

\bibitem{chen2016quickest}
Chen, Y.C., Banerjee, T., Dom{\'\i}nguez.Garc{\'\i}a, A.D., Veeravalli, V.V.:
  `{Quickest Line Outage Detection and Identification}', \emph{IEEE
  Transactions on Power Systems},  2016, \textbf{31}, (1), pp.~749--758

\bibitem{deka2016estimating}
Deka, D., Backhaus, S., Chertkov, M.
\newblock `{Estimating Distribution Grid Topologies: A Graphical Learning based
  Approach}'.
\newblock In: Power Systems Computation Conference. (IEEE,  2016. pp.~ 1--7

\bibitem{cover2012elements}
Cover, T.M., Thomas, J.A.: `{Elements of Information Theory}'.
\newblock (John Wiley \& Sons,  2012)

\bibitem{deka2017topology}
Deka, D., Talukdar, S., Chertkov, M., Salapaka, M.: `{Topology Estimation in
  Bulk Power Grids: Guarantees on Exact Recovery}', \emph{arXiv preprint
  arXiv:170701596},  2017,

\bibitem{chow1968approximating}
Chow, C., Liu, C.: `{Approximating Discrete Probability Distributions with
  Dependence Trees}', \emph{IEEE Transactions on Information Theory},  1968,
  \textbf{14}, (3), pp.~462--467

\bibitem{wang2016phase}
Wang, W., Yu, N., Foggo, B., Davis, J., Li, J.
\newblock `{Phase Identification in Electric Power Distribution Systems by
  Clustering of Smart Meter Data}'.
\newblock In: Machine Learning and Applications (ICMLA), 2016 15th IEEE
  International Conference on. (IEEE,  2016. pp.~ 259--265

\bibitem{short2013advanced}
Short, T.A.: `{Advanced Metering for Phase Identification, Transformer
  Identification, and Secondary Modeling}', \emph{IEEE Transactions on Smart
  Grid},  2013, \textbf{4}, (2), pp.~651--658

\bibitem{arya2011phase}
Arya, V., Seetharam, D., Kalyanaraman, S., Dontas, K., Pavlovski, C., Hoy, S.,
  et~al.
\newblock `{Phase Identification in Smart Grids}'.
\newblock In: Smart Grid Communications (SmartGridComm), 2011 IEEE
  International Conference on. (IEEE,  2011. pp.~ 25--30

\bibitem{kruskal1956shortest}
Kruskal, J.B.: `{On the Shortest Spanning Subtree of a Graph and the Traveling
  Salesman Problem}', \emph{Proceedings of the American Mathematical society},
  1956, \textbf{7}, (1), pp.~48--50

\bibitem{cormen2001introduction}
Cormen, T.H., Leiserson, C.E., Rivest, R.L., Stein, C., et~al.: `{Introduction
  to Algorithms}'. vol.~2.
\newblock (MIT press Cambridge,  2001)

\bibitem{liao2015distribution}
Liao, Y., Weng, Y., Wu, M., Rajagopal, R.
\newblock `Distribution grid topology reconstruction: An information theoretic
  approach'.
\newblock In: North American Power Symposium. (,  2015. pp.~ 1--6

\bibitem{cavraro2019voltage}
Cavraro, G., Kekatos, V., Veeramachaneni, S.: `Voltage analytics for power
  distribution network topology verification', \emph{IEEE Transactions on Smart
  Grid},  2019, \textbf{10}, (1), pp.~1058--1067

\bibitem{baran1989network}
Baran, M.E., Wu, F.F.: `{Network Reconfiguration in Distribution Systems for
  Loss Reduction and Load Balancing}', \emph{IEEE Transactions on Power
  delivery},  1989, \textbf{4}, (2), pp.~1401--1407

\bibitem{dugan2010ieee}
Dugan, R., Arritt, R.: `The ieee 8500-node test feeder', \emph{Electric Power
  Research Institute, Palo Alto, CA, USA},  2010,

\bibitem{zhao2018learning}
Zhao, Y., Chen, J., Poor, H.V.
\newblock `{Learning to Infer Power Grid Topologies: Performance and
  Scalability}'.
\newblock In: 2018 IEEE Data Science Workshop (DSW). (IEEE,  2018. pp.~
  215--219

\bibitem{dobos2014pvwatts}
Dobos, A.P.: `{PVWatts Version 5 Manual}', \emph{National Renewable Energy
  Laboratory, September},  2014,

\bibitem{ellis2012review}
Ellis, A., Nelson, R., Von.Engeln, E., Walling, R., MacDowell, J., Casey, L.,
  et~al.
\newblock `{Review of Existing Reactive Power Requirements for Variable
  Generation}'.
\newblock In: Power and Energy Society General Meeting, 2012 IEEE. (IEEE,
  2012. pp.~ 1--7

\bibitem{ieee2014guide}
IEEE: `{IEEE Guide for Conducting Distribution Impact Studies for Distributed
  Resource Interconnection}', \emph{IEEE Std 15477-2013},  2014, pp.~ 1--137

\bibitem{jabr2014minimum}
Jabr, R.A.: `{Minimum Loss Operation of Distribution Networks with Photovoltaic
  Generation}', \emph{IET Renewable Power Generation},  2014, \textbf{8}, (1),
  pp.~33--44

\bibitem{dorostkar2016value}
Dorostkar.Ghamsari, M.R., Fotuhi.Firuzabad, M., Lehtonen, M., Safdarian, A.:
  `{Value of Distribution Network Reconfiguration in Presence of Renewable
  Energy Resources}', \emph{IEEE Transactions on Power Systems},  2016,
  \textbf{31}, (3), pp.~1879--1888

\bibitem{zheng2013smart}
Zheng, J., Gao, D.W., Lin, L.
\newblock `{Smart Meters in Smart Grid: An Overview}'.
\newblock In: IEEE Green Technologies Conference. (,  2013. pp.~ 57--64

\bibitem{ansc12}
{National Electrical Manufacturers Association}.
\newblock `{ANSI C12.20-2010: American National Standard for Electricity Meter:
  0.2 and 0.5 Accuracy Classes}'.
\newblock (American National Standards Institute,  2010.

\bibitem{chassin2008gridlabd}
{Chassin}, D.P., {Schneider}, K., {Gerkensmeyer}, C.
\newblock `Gridlab-d: An open-source power systems modeling and simulation
  environment'.
\newblock In: 2008 IEEE/PES Transmission and Distribution Conference and
  Exposition. (,  2008. pp.~ 1--5

\end{thebibliography}

\end{document}